\documentclass[12pt]{article}

\usepackage[T1]{fontenc}
\usepackage{kpfonts}

\usepackage{setspace}
\usepackage[margin=1.25in]{geometry}

\usepackage[dvipsnames]{xcolor}

\usepackage{pdfpages}
\usepackage{float}
\usepackage{multirow}
\usepackage{caption}
\usepackage{subcaption}
\usepackage{epigraph}

\usepackage{mathtools}
\usepackage{amsthm}
\usepackage{aliascnt} % lets cleveref distinguish theorem-like environments
\usepackage{accents}
\usepackage{dutchcal}
\usepackage{bbm}

\usepackage{enumitem}
\usepackage{sgame}

\usepackage{tikz}
\usepackage{tikz-cd}
\usetikzlibrary{calc,shapes,arrows}

\usepackage{tcolorbox}
\usepackage{listings}

\usepackage[normalem]{ulem}

\usepackage[round]{natbib}
\newcommand{\ubar}[1]{\underaccent{\bar}{#1}}

\DeclareMathOperator{\supp}{supp}

\DeclareMathOperator{\co}{co}

\newcommand{\R}{\mathbb{R}}

\newtheorem{theorem}{Theorem}[section]

\newaliascnt{proposition}{theorem}
\newtheorem{proposition}[proposition]{Proposition}
\aliascntresetthe{proposition}

\newaliascnt{lemma}{theorem}
\newtheorem{lemma}[lemma]{Lemma}
\aliascntresetthe{lemma}

\newaliascnt{corollary}{theorem}
\newtheorem{corollary}[corollary]{Corollary}
\aliascntresetthe{corollary}

\newaliascnt{claim}{theorem}

\aliascntresetthe{claim}

\theoremstyle{definition}

\newaliascnt{definition}{theorem}

\aliascntresetthe{definition}

\newaliascnt{example}{theorem}

\aliascntresetthe{example}

\newaliascnt{assumption}{theorem}
\newtheorem{assumption}[assumption]{Assumption}
\aliascntresetthe{assumption}

\newaliascnt{condition}{theorem}

\aliascntresetthe{condition}

\newaliascnt{question}{theorem}

\aliascntresetthe{question}

\newaliascnt{remark}{theorem}

\aliascntresetthe{remark}

\newaliascnt{remarks}{theorem}

\aliascntresetthe{remarks}

\newaliascnt{aside}{theorem}

\aliascntresetthe{aside}

\newaliascnt{note}{theorem}

\aliascntresetthe{note}

\usepackage{thmtools}
\usepackage{thm-restate}

\usepackage{hyperref}

\hypersetup{
    colorlinks=true,
    linkcolor=OrangeRed,
    filecolor=Thistle,
    urlcolor=Thistle,
    citecolor=Thistle,
}

\usepackage[nameinlink]{cleveref}

\crefname{theorem}{theorem}{theorems}
\Crefname{theorem}{Theorem}{Theorems}

\crefname{proposition}{proposition}{propositions}
\Crefname{proposition}{Proposition}{Propositions}

\crefname{lemma}{lemma}{lemmas}
\Crefname{lemma}{Lemma}{Lemmas}

\crefname{corollary}{corollary}{corollaries}
\Crefname{corollary}{Corollary}{Corollaries}

\crefname{claim}{claim}{claims}
\Crefname{claim}{Claim}{Claims}

\crefname{definition}{definition}{definitions}
\Crefname{definition}{Definition}{Definitions}

\crefname{example}{example}{examples}
\Crefname{example}{Example}{Examples}

\crefname{assumption}{assumption}{assumptions}
\Crefname{assumption}{Assumption}{Assumptions}

\makeatletter
\let\cref@old@isrefconsecutive\cref@isrefconsecutive

\def\cref@isrefconsecutive#1#2{%
  \begingroup
    \def\cref@assumptiontype{assumption}%
    \cref@gettype{#1}{\cref@typea}%
    \ifx\cref@typea\cref@assumptiontype
      \endgroup
      \@cref@refconsecutivefalse
    \else
      \endgroup
      \cref@old@isrefconsecutive{#1}{#2}%
    \fi
}
\makeatother

\crefname{condition}{condition}{conditions}
\Crefname{condition}{Condition}{Conditions}

\crefname{question}{question}{questions}
\Crefname{question}{Question}{Questions}

\crefname{remark}{remark}{remarks}
\Crefname{remark}{Remark}{Remarks}

\crefname{remarks}{remarks}{remarks}
\Crefname{remarks}{Remarks}{Remarks}

\crefname{aside}{aside}{asides}
\Crefname{aside}{Aside}{Asides}

\crefname{note}{note}{notes}
\Crefname{note}{Note}{Notes}

\crefname{appendix}{appendix}{appendices}
\Crefname{appendix}{Appendix}{Appendices}

\newcommand{\secref}[1]{\hyperref[#1]{\S\ref*{#1}}}

\definecolor{backcolour}{rgb}{0.63, 0.79, 0.95}

\lstdefinestyle{mystyle}{
  backgroundcolor=\color{backcolour},
  basicstyle=\ttfamily\footnotesize,
  breakatwhitespace=false,
  breaklines=true,
  captionpos=b,
  keepspaces=true,
  numbers=left,
  numbersep=5pt,
  showspaces=false,
  showstringspaces=false,
  showtabs=false,
  tabsize=2
}

\begin{document}
\title{Distributional Competition}
\author{Mark Whitmeyer\thanks{Arizona State University. Email: \href{mailto:mark.whitmeyer@gmail.com}{mark.whitmeyer@gmail.com}. Dedicated to KS. Both robots (\href{https://www.refine.ink/}{refine.ink}, and ChatGPT 4.-6.) and humans (Ralph Boleslavsky, Teddy Kim, R. Vijay Krishna, Doron Ravid, and Joseph Whitmeyer) helped me improve this paper.}}

\date{\today}

\maketitle

\begin{abstract}
I study symmetric competitions in which each player chooses an arbitrary distribution over a performance index, subject to a convex cost. I establish existence of a symmetric equilibrium, document various properties it must possess, and provide a characterization via the first-order approach. Several applications--to R\&D competition, oligopolistic competition with product design, and rank-order contests--follow.
\end{abstract}

\section{Introduction}\label{sec:intro}

I study a class of strategic interactions in which each player chooses a probability distribution over a one-dimensional performance index, incurring a convex cost for generating that distribution. Payoffs depend on the realized performance profile, in a way that rewards higher realizations and is symmetric across players. Standard rank-order tournament and contest models fit as special cases, but my framework is broader, subsuming a wide range of competitive problems with flexible production and endogenous risk-taking.

My framework is agnostic vis-\`{a}-vis microfoundations. Rather than deriving the feasible set of outcome distributions from, for instance, a particular underlying stochastic process (a diffusion, a Poisson-arrival model, etc.) or a parametrized technology, I allow players to choose distributions directly and fully flexibly, and encode technological constraints through a general cost functional over distributions. This accommodates different underpinnings across applications while keeping the strategic and welfare logic transparent.

The special case of my model in which costs are linear in the distribution is precisely the standard setting, in which choosing a distribution is equivalent to randomizing over deterministic performance levels and paying the expected cost. The environments I highlight allow costs to depend nonlinearly on the distribution itself: on aggregate performance intensity, interactions among different parts of the distribution, or the probability of reaching success by different deadlines and cannot be reduced to ordinary mixed strategies over pure (scalar) actions. They generate new comparative statics because changing a distribution changes not only the probability of winning but also the marginal cost of probability mass elsewhere in the distribution.

The generality of my framework allows me to accommodate a number of scenarios. In risky R\&D, a firm’s choice of project mix and experimentation intensity begets a distribution over discovery times. How does competition affect research intensity and output? How does competition affect quality choices when firms choose prices as well; must prices converge to marginal cost in the perfect competition limit? In rank-order contests, how do more inegalitarian prize schedules shape output?

I begin my analysis by using results in \citet{reny1999existence} to prove existence of a symmetric pure-strategy equilibrium, even though discontinuities arise naturally due to ties generating jumps in payoffs. Second, I show that under natural conditions, symmetric equilibria must take a particular atomless form. Third, I characterize both symmetric equilibria and the corresponding symmetric planner’s problem via the first-order approach

I then visit various particular settings of interest. I tackle the classic question on the relationship between prize inequality and effort in rank-order contests. For a broad class of costs, I derive a stark conclusion: more inegalitarian prizes lead to higher mean but riskier output (formally, a shift in the increasing convex order). I also show that in risky R\&D competition, equilibrium leads to inefficiently high effort compared to the planner's problem, clarifying and elucidating classic results. Breakthroughs arise inefficiently quickly: the equilibrium distribution is first-order stochastically dominated (FOSD) by the  planner's solution. Under nonlinear deadline costs, this acceleration endogenously raises the marginal cost of further acceleration, a force absent under linear expected costs.

Finally, I study competition between oligopolistic firms who choose both prices and (stochastic) quality. My goal here is to ask--when quality is endogenous--when prices converge to marginal cost as the number of firms grows large. I provide an intuitive necessary condition for this to occur. Notably, for distributional-costs that satisfy a natural steepness property on their derivative, convergence to marginal cost implies no product differentiation. That is, the Bertrand logic is flipped: marginal-cost pricing \textit{implies} perfectly homogeneous products.\footnote{Here again, my framework engenders an insight absent from the linear-cost  environment. In particular, it endogenizes the marginal-cost gap over the equilibrium  quality range. With linear costs, this gap is exogenous: it is fixed by the cost schedule  for higher versus lower quality. With nonlinear costs, the marginal-cost schedule  depends on the whole product-design distribution; thus, in the competitive limit,  the equilibrium distribution itself must make marginal costs nearly flat over the  qualities firms produce.}

\smallskip

\noindent \textbf{Roadmap.} \secref{sec:lit} completes the first section with a discussion of related papers. \secref{sec:env} introduces the general environment. The general results lie in \secref{sec:results}. There, I prove equilibrium existence and characterize the basic structure. In \secref{sec:roco}, I specialize the analysis to rank-order contests. \secref{sec:rd} explores risky R\&D, and \secref{sec:unit_demand_prices} scrutinizes price and quality competition between firms. Omitted proofs lie in \Cref{app:technical-details,app:b,app:kkt,ap:roproofs,rdproofs}.

\subsection{Related Literature}\label{sec:lit}
\label{sec:related}
This paper is closely connected to the literature on risk-taking in contests in which agents flexibly choose distributions with a fixed mean.\footnote{See, in particular, \citet{BoleslavskyCotton2015, BoleslavskyCotton2018,FangNoe2016,wagman2012choosing,albrecht,AU2,spiegler,KimKoh2022,au2023attraction}.} Via the Skorokhod embedding theorem, this problem is mathematically equivalent to contests in which the contestants run diffusions.\footnote{This literature originates with \citet{SeelStrack2013} and includes \citet{FengHobson2015,feng2016gambling,FengHobson2016,NutzZhang2022,whitmeyer2023submission}.} My analysis is complementary: although choosing risk subject to a mean constraint is natural in many settings, it is also reasonable to model many scenarios as agents making an arbitrary (but costly) choice of stochastic output. Further afield are the (older) papers that study risk-taking contests more broadly.\footnote{A non-exhaustive list is \cite{dasgupta1980uncertainty, bhattacharya1986portfolio,klette1986market}, who study R\&D contests (which I revisit in \secref{sec:rd}); \cite{hvide2002tournament,hvide2003risk,goel2008overconfidence,gilpatric2009risk,fang2021less}, who analyze promotion contests; \cite{basak2015competition,strack2016risk,whitmeyer2019relative,lacker2019mean}, who investigate investment managers who care about others; and \cite{robson1992status,hopkins2018inequality}, where relative \textit{status} is the objective.}

Naturally, the classic all-pay contest literature--in which agents choose sunk bids or efforts and prizes are allocated by rank, with expenditures being paid regardless of whether one wins--is also related to my exercise.\footnote{See, e.g., \citet{HillmanSamet1987,HillmanRiley1989,BayeKovenockDeVries1996,BarutKovenock1998,AmannLeininger1996,CheGale2000,FangNoeStrack2020}.} Viewed through the lens of this paper, once agents are allowed to mix, the strategic object in the canonical all-pay model is a distribution over bids/efforts, and the standard expenditure technology implies that expected costs enter as a linear functional of that distribution (\textit{viz.}, \(C(F)=\int cdF\)). Thus, the standard all-pay environments are subsumed by my contest application as a special case of linear distributional costs.

The paper closest to this one, \citet{KimKrishnaRyvkin2023}, combines costly effort with an endogenous choice of mean-preserving noise (``choosing your own luck'') and shows that allowing strategic risk-taking can dramatically reshape standard contest design conclusions. Following my rank-order contest application (\secref{sec:roco}) I elaborate on the contrast between their paper and mine. \citet{HwangKimBoleslavsky2023} study the joint information-provision and pricing problem of oligopolists, which is formally a problem in which firms price and flexibly choose distributions according to a mean-preserving contraction constraint. My joint pricing and quality choice application (\secref{sec:unit_demand_prices}) looks at a related problem.

Also related is \citet{QuintKojima2026}, who study symmetric equilibria in pre-auction investment. In their model, bidders make covert investments that determine their valuation distributions before an auction. They show that, for a broad class of auction formats, the symmetric investment equilibrium exists, is essentially unique, and is invariant across auction formats. Their environment is closer to auction theory and exploits the structure of pre-auction investment; mine instead studies a general class of competitions.

Methodologically, this paper builds most directly on \citet{GeorgiadisRavidSzentes2024}, who develop a generalized first-order approach for flexible moral hazard problems in which a single agent chooses an output distribution at a smooth cost. \citet{ravid2022learning} look at ``learning before trading,'' in which a buyer's cost of learning about her value takes this general smooth form. A trio of papers by \citeauthor{Kraehmer2024} look at hold up, security design, and monopoly regulation with costly distributional choices.\footnote{See \citet{Kraehmer2024,Kraehmer2025reg,Kraehmer2025sec}.} \citet{CastroPiresKattwinkelKnoepfle2025} extend the flexible-distribution approach to environments with adverse selection.

\section{Formal Environment}\label{sec:env}

There are \(n\ge 2\) players. Let \(X\coloneqq \left[0,1\right]\), endowed with its usual compact metric topology. We interpret \(x\in X\) as a one-dimensional performance index, with higher \(x\) being better. Each player's strategy is a cumulative distribution function (cdf), which is a right-continuous, nondecreasing function \(F\colon \left[0,1\right]\to\left[0,1\right]\) with \(F(1) = 1\). \(\mathcal F\) denotes the set of all such \(F\). Each \(F\in\mathcal F\) induces a unique probability measure \(dF\) on \(X\), and we write \(\supp(dF)\subseteq X\) for its support.

Given a profile \((F_1,\ldots,F_n)\in\mathcal F^n\), the realized performances \(X_i\sim dF_i\) are drawn independently. We break ties uniformly. Player \(i\)'s gross payoff \(\pi_i\colon X^n\to\R\) denotes her expected payoff conditional on the realized performance profile, after applying the uniform tie-breaking rule. We explicitly detail this object in \Cref{app:technical-details}. We make the following assumptions concerning the gross payoff:
\begin{assumption}\label{ass:prize}
\hfill
\begin{enumerate}[noitemsep,label=(P\arabic*),leftmargin=*]
\item \label{ass:prize-bdd} Bounded payoffs: there exists \(\bar{\pi}<\infty\) such that \(0\le \pi_i\le \bar{\pi}\) for all \(i\).
\item \label{ass:prize-sym} The game is symmetric: for any permutation \(\sigma\) of \(\{1,\ldots,n\}\),
\[\pi_{\sigma(i)}\left(x_{\sigma(1)},\ldots,x_{\sigma(n)}\right) = \pi_i\left(x_1,\ldots,x_n\right).\]
\item \label{ass:prize-disc} Discontinuities only arise at ties: for each \(i\), \(\pi_i\) is continuous at every \(x\in X^n\) such that \(x_i\ne x_j\) for all \(j\ne i\).
\item \label{ass:prize-mono} Higher is better: for each \(i\), every \(x_{-i}\), and all \(x_i'\ge x_i\), \(\pi_i\left(x_i',x_{-i}\right)\ge \pi_i\left(x_i,x_{-i}\right)\).
\item \label{ass:prize-agg} Continuous aggregate prize: there exists a bounded continuous symmetric function \(\Pi\colon X^n\to\R\) such that for all \(x\in X^n\), \(\sum_{i=1}^n \pi_i\left(x_1,\ldots,x_n\right)=\Pi(x_1,\ldots,x_n)\).
\end{enumerate}
\end{assumption}

If player \(i\) chooses the cdf \(F_i\), she incurs the cost \(C\left(F_i\right)\). We stipulate that the players' net payoffs are additively separable in the gross payoff and the cost and assume the following about the cost \(C\):
\begin{assumption}\label{ass:cost} The cost functional \(C\colon\mathcal F\to\mathbb R\) satisfies the following conditions.
\begin{enumerate}[noitemsep,label=(C\arabic*),leftmargin=*]
\item \label{ass1} \(C\) is convex and continuous with respect to weak convergence of the associated measures on \(X\).
\item \label{ass2} \(C\) has a continuous G\^{a}teaux derivative in the following sense: for each \(F\in\mathcal{F}\) there exists a function \(c_F \colon X\to\R\), continuous in \(x\), such that for all \(G\in\mathcal{F}\),
\[\delta C\left(F;G-F\right)\coloneqq\lim_{\varepsilon\downarrow0}\frac{C\left(\left(1-\varepsilon\right)F+\varepsilon G\right)-C\left(F\right)}{\varepsilon}\]
exists and satisfies \(\delta C\left(F;G-F\right)=\int_X c_F(x)d\left(G-F\right)(x)\).
\item \label{ass3} There exists \(\eta_1>0\) such that for every \(F\in\mathcal{F}\), \(c_F(1)\ge c_F(0)+\bar{\pi}+\eta_1\).
\end{enumerate}
\end{assumption}
As \(G-F\) has total mass zero, \(c_F\) is unique only up to an \(F\)-dependent additive constant. All first-order and equilibrium conditions below are invariant to this normalization. \Cref{ass:cost}\ref{ass2} is stronger than simply requiring directional derivatives to exist: it requires the derivative to be represented by integration against a continuous function \(c_F\).

Under profile \((F_1,\ldots,F_n)\in\mathcal F^n\), player \(i\)'s net expected payoff is
\[u_i\left(F_1,\ldots,F_n\right) \coloneqq \mathbb E\left[\pi_i\left(X_1,\ldots,X_n\right)\right]-C\left(F_i\right),\]
where \(X_j\sim dF_j\) independently for all \(j=1,\ldots,n\).
\subsection{Discussion of Assumptions}
\Cref{ass:prize} places minimal structure on the environment. Boundedness \ref{ass:prize-bdd} is a simple regularity condition that guarantees expected payoffs are well-defined for any distributional choices. Symmetry \ref{ass:prize-sym} is simply our specialization to symmetric environments. The key modeling restriction is \ref{ass:prize-disc}: a player’s payoff can jump only when her realization is exactly tied with someone else's. This fits our leading applications (rank-order contests, product competition, priority races), where discontinuities arise precisely because a tie changes who wins. Monotonicity \ref{ass:prize-mono} formalizes the ``higher performance is better'' interpretation of the index and rules out perverse cases where improving one’s realization could reduce one's prize. Finally, \ref{ass:prize-agg} is another form of regularity, stipulating that the total prize to be allocated depends continuously on the realized performance profile (and not, e.g., on the auxiliary tie-breaking rule).

\Cref{ass:cost} spells out the production technology. Convexity and continuity \ref{ass1} capture the idea that pushing probability mass toward desirable outcomes exhibits weakly increasing marginal cost and that small changes in the outcome distribution do not produce large, discontinuous swings in costs. Differentiability \ref{ass2} provides a clean marginal-cost \(c_F\). It is the first-order cost of locally shifting weight, and it is what ultimately facilitates the first-order approach. The last, \ref{ass3}, is vaguely an Inada condition, guaranteeing that placing mass at the maximal performance level is too expensive to ever be optimal. This rules out corner solutions at the upper bound and the pathologies that might follow.

\subsection{Examples of Costs}\label{sec:exes}

We now briefly discuss costs that satisfy \Cref{ass:cost}\ref{ass1} and \ref{ass2}. \Cref{ass:cost}\ref{ass3} is a separate boundary condition.

\smallskip

\noindent \textbf{Linear.} Let \(\gamma\colon X\to\mathbb R\) be continuous and define \(C\left(F\right)=\int_X\gamma(x)dF(x)\). Then \(C\) is affine, hence, convex, and
\[\delta C\left(F;G-F\right)=\int_X\gamma(x)d\left(G-F\right)(x),\]
so we may take \(c_F(x)=\gamma(x)\). This is the standard expected-cost benchmark. For example, in rank-order contests it is the usual all-pay specification where players mix over deterministic efforts.

\smallskip

\noindent \textbf{Single-index.} Again let \(\gamma\colon X\to\mathbb R\) be continuous and let \(\Upsilon \colon \mathbb{R} \to \mathbb{R}\) be increasing, convex, and continuously differentiable on an interval containing \(\gamma\left(X\right)\). Let \(C\left(F\right)=\Upsilon\left(\int_X\gamma(x)dF(x)\right)\), so that
\[\delta C\left(F;G-F\right)=\Upsilon'\left(\int_X\gamma(x)dF(x)\right)\int_X\gamma(x)d\left(G-F\right)(x),\]
so we may take \(c_F(x)=\Upsilon'\left(\int_X\gamma(x)dF(x)\right)\gamma(x)\). This nests the linear case when \(\Upsilon\) is affine and captures increasing marginal costs of aggregate performance intensity when \(\Upsilon\) is strictly convex. We specialize to this family in \secref{sec:roco}.

\smallskip

\noindent \textbf{Symmetric-interaction.} Let \(\gamma\colon X\to\mathbb R\) be continuous and let \(K\colon X\times X\to\mathbb R\) be continuous, symmetric, and positive semidefinite.\footnote{That is, \(\int_X\int_XK(x,y)d\nu(x)d\nu(y)\ge0\) for every finite signed measure \(\nu\) on \(X\).} Define
\[C\left(F\right)=\int_X\gamma(x)dF(x)+\frac{1}{2}\int_X\int_XK(x,y)dF(x)dF(y).\]
Then \(C\) is convex and
\[\delta C\left(F;G-F\right)=\int_X\left(\gamma(x)+\int_XK(x,y)dF(y)\right)d\left(G-F\right)(x),\]
so we may take \(c_F(x)=\gamma(x)+\int_XK(x,y)dF(y)\).

\smallskip

\noindent \textbf{Cdf-path.} Let \(\eta\) be a finite non-atomic Borel measure on \(X\), \(\Lambda\colon X\times\left[0,1\right]\to\mathbb R\) be continuous, and \(q\mapsto\Lambda(s,q)\) be convex and continuously differentiable for every \(s\). Suppose also that \(\Lambda_q\) is bounded and continuous. Defining \(C\left(F\right)=\int_X\Lambda\left(s,F(s)\right)d\eta(s)\),
\[\delta C\left(F;G-F\right)=\int_X\Lambda_q\left(s,F(s)\right)\left(G(s)-F(s)\right)d\eta(s)=\int_X\left(\int_{\left[x,1\right]}\Lambda_q\left(s,F(s)\right)d\eta(s)\right)d\left(G-F\right)(x),\]
and so we may take \(c_F(x) = \int_{\left[x,1\right]} \Lambda_q\left(s,F(s)\right) d\eta(s)\). We specialize to this family in \secref{sec:rd}.

These examples clarify the distinction between standard mixed-strategy competition and my more general framework. Linear costs reduce the model to expected-cost randomization over deterministic actions. Single-index costs make the marginal cost of every realization depend on an endogenous aggregate intensity index. Interaction costs make the cost of probability mass at one outcome depend on where the rest of the distribution lies. Cdf-path costs are natural in deadline problems: probability mass at an early time raises the probability of success by every later deadline, so the marginal cost of that mass incorporates these spillovers.

\section{Results}\label{sec:results}

\subsection{Preliminary Results}
We begin by establishing that a symmetric (pure-strategy) equilibrium exists, deferring the proof to \Cref{existenceproof}.\footnote{In the Supplementary Appendix, we also prove symmetric equilibrium existence for various subsets of \(\mathcal{F}\) (like those generated by mean or majorization constraints).}
\begin{proposition}\label{prop:symNE}
Under \Cref{ass:prize,ass:cost}, there exists a symmetric pure-strategy Nash equilibrium \(F^*\in\mathcal{F}\).
\end{proposition}
A central obstacle in these games is that payoffs jump discontinuously at ties: when two or more players draw exactly the same performance level, ``who wins'' changes abruptly. In spite of these discontinuities, a symmetric pure-strategy equilibrium exists. The intuition is that the discontinuities are tightly structured. By construction, they come only from tie events, and those tie events can be made essentially negligible by slightly ``smoothing'' deviations. Once deviations are smoothed in this way, the deviator's expected payoff varies continuously with opponents' strategies in a neighborhood, which is exactly the kind of stability we need for an equilibrium existence argument in this discontinuous game.

We next establish further properties that symmetric equilibria must possess, with proofs left to \Cref{app:b}. First, we argue that no best response places mass at \(1\), and so neither does any symmetric equilibrium.
\begin{lemma}\label{lem:no-atom-1}
Posit \Cref{ass:cost}. For any symmetric equilibrium \(F\in\mathcal F\), \(dF(\{1\})=0\).
\end{lemma}

Second, we introduce an additional assumption to preclude atoms on \(\left(0,1\right)\). To elaborate, atomlessness on \(\left(0,1\right)\) generally requires a strict tie disadvantage: whenever a player is tied for the best performance at a given interior level, moving slightly above it must yield a strictly higher expected prize on that tie event. This holds in our later specializations, and we record an intuitive condition for the abstract model.

\begin{assumption}\label{ass:strict-tie}
For each \(x_0\in\left(0,1\right)\) there exist numbers \(\delta(x_0)\in\left(0,1-x_0\right]\) and \(\Delta(x_0)>0\) such that for every \(z=\left(x_2,\ldots,x_n\right)\) satisfying \(\max_{2\le j\le n}x_j=x_0\), we have, for every \(\varepsilon\in\left(0,\delta(x_0)\right]\), \(\pi_1\left(x_0+\varepsilon;z\right) \geq \pi_1\left(x_0;z\right) + \Delta(x_0)\).
\end{assumption}
In short, this assumption introduces an ``over-cutting'' force, typical in many settings of interest: a player can get a discrete jump in her gross payoff by escaping a tie. This allows us to eliminate interior atoms in symmetric equilibria:
\begin{lemma}\label{lem:no-interior-atoms}
Posit \Cref{ass:prize,ass:cost,ass:strict-tie}. If \(F\in\mathcal F\) is a symmetric equilibrium, then \(dF(\{x\})=0\) for every \(x\in\left(0,1\right)\).
\end{lemma}

Now we establish a useful continuity property of a player's interim payoff. For \(x\in X\), define player \(1\)'s interim expected prize against symmetric opponents \(F\) as
\[a_F(x)\coloneqq \mathbb E\left[\pi_1\left(x,X_2,\ldots,X_n\right)\right],\]
where \(X_2,\ldots,X_n\sim dF\) are i.i.d. Then, for any deviation \(H\in\mathcal F\),
\[u_1\left(H,F,\ldots,F\right) = \int_X a_F(x) dH(x)-C(H).\]
Our next lemma records a continuity property of \(a_F\) at points where opponents place no atom:

\begin{lemma}\label{lem:aF-cont}
Posit \Cref{ass:prize,ass:cost}. If \(x_0\in X\) satisfies \(dF(\{x_0\})=0\), then \(a_F(\cdot)\) is continuous at \(x_0\).
\end{lemma}

Next we turn our attention to the symmetric planner problem in which the planner chooses \(F\in\mathcal F\) to maximize the diagonal payoff
\[v(F)\coloneqq u_1\left(F,\ldots,F\right).\]
Recall that in \Cref{ass:prize}, we introduce the notation
\[\sum_{i=1}^n \pi_i\left(x_1,\ldots,x_n\right)=\Pi(x_1,\ldots,x_n).\]
Thus, if \(X_1,\ldots,X_n\sim dF\) are i.i.d., the planner's objective is
\[v(F)=\frac{1}{n}\mathbb E\left[\Pi(X_1,\ldots,X_n)\right]-C(F).\]
As we show in the appendix \Cref{lem:vcont}, \(v\) is continuous on the compact metric space \(\mathcal F\), and, therefore, it attains a maximum:
\begin{proposition}\label{prop:planner-exists}
Under \Cref{ass:prize,ass:cost}, there exists \(G\in\mathcal F\) maximizing \(v\) over \(\mathcal F\).
\end{proposition}

\subsection{Equilibria and Planner Optima Characterizations}\label{sec:kkt}
We now characterize symmetric equilibria and a necessary condition for planner optima via the first-order approach, leaving the technical details for \Cref{app:kkt}. Defining the net (symmetric) payoff
\[\Phi(F,x)\coloneqq a_F(x)-c_F(x),\]
\begin{proposition}\label{lem:game-kkt}
Posit Assumptions~\ref{ass:prize} and \ref{ass:cost}. A distribution \(F\in\mathcal F\) is a symmetric equilibrium if and only if there exists a multiplier \(\lambda_g\in\R\) such that
\[\Phi(F,x)\le \lambda_g \ \forall x\in X, \qquad \text{and} \qquad \Phi(F,x)=\lambda_g \ \ \forall x\in \supp(dF).\]
\end{proposition}

We now turn our attention to the planner. For \(F\in\mathcal F\) and \(x\in X\), define the planner's \textit{interim aggregate prize}
\[A_F(x)\coloneqq \mathbb E\left[\Pi\left(x,X_2,\ldots,X_n\right)\right],\]
where \(X_2,\ldots,X_n\sim dF\) are i.i.d.

\begin{proposition}\label{lem:planner-kkt}
Posit \Cref{ass:prize,ass:cost}. If \(G\in\mathcal F\) maximizes \(v\) over \(\mathcal F\), then there exists \(\lambda_p\in\R\) such that
\[A_G(x)-c_G(x)\le \lambda_p \ \ \forall x\in X, \quad \text{and} \quad A_G(x)-c_G(x)=\lambda_p \ \ \forall x\in \supp(dG).\]
If \(v\) is concave on \(\mathcal{F}\), these conditions also characterize the maximizers of \(v\), and if \(v\) is strictly concave, its maximizer is unique.
\end{proposition}

The economics of equilibrium and the planner's optimum are basic. Namely, \Cref{lem:game-kkt,lem:planner-kkt} are intuitive generalizations of the standard first-order-approach logic to a setting where the decision variable is not a single number but a distribution. The key idea is to think about a tiny reallocation of probability mass: take an \(\varepsilon\)-slice of mass away from one outcome level and move it to another. If such a local reallocation would raise expected payoff, then the original distribution cannot be optimal. Therefore, at an optimum, every outcome level in the support of the optimizer must deliver the same net marginal return, and any outcome level outside the support can deliver no more.

Economically, \Cref{lem:game-kkt} is exactly the familiar indifference logic of equilibria. The player only puts weight on outcomes that maximize private net marginal return, and it must be indifferent (in net marginal terms) across all outcomes it actually uses. The multiplier is the common ``benchmark'' net marginal return in equilibrium, which arises because probability mass is a scarce resource: adding weight at one \(x\) necessarily means removing weight elsewhere.

The planner's problem has the same structure, but the benefit side changes from private to social. When the planner shifts probability mass toward \(x\), the relevant marginal benefit is the effect on total surplus (the aggregate prize). The marginal cost of shifting probability mass toward \(x\) is the same as in the game, and the planner uses exactly the same ``equalize net marginal returns on the support'' principle, but with social marginal benefits instead of private ones.

\section{Specialization I: Rank-Order Contests}\label{sec:roco}
We begin by studying rank-order contests. For \(n \geq 2\), we specialize \(\pi_i\) to a rank-order contest with prize vector \(\mathbf{v} = \left( v_1,\ldots,v_n \right)\) satisfying
\[v_1 \ge \cdots \ge v_n = 0, \qquad\text{and}\qquad \sum_{k=1}^n v_k = 1,\]
with ties broken uniformly.

We specialize to a particular family of costs:
\begin{assumption}\label{ass:contest}
The cost functional satisfies
\[C\left(F\right) = \Upsilon\left(\int_0^1 \gamma\left(x\right)dF\left(x\right)\right), \qquad\text{for all \(F \in \mathcal{F}\),}\]
where \(\gamma\colon \left[0,1\right] \to\mathbb{R}_+\) and \(\Upsilon\colon\left[0,\gamma\left(1\right)\right]\to\mathbb R\) satisfy:
\begin{enumerate}[noitemsep,leftmargin=*]
\item \(\gamma\) is continuous, strictly increasing, concave, and \(\gamma\left( 0 \right)=0\).
\item \(\Upsilon\) is increasing, convex, and continuously differentiable.
\item \(z \mapsto z\Upsilon'\left(z\right)\) is strictly increasing on \(\left[0,\gamma\left(1\right)\right]\).
\item \(\Upsilon'\left(0\right)\gamma\left(1\right)>1\).
\end{enumerate}
\end{assumption}
Define \(z_F \coloneqq \int_0^1 \gamma\left(x\right)dF\left(x\right)\). Under \Cref{ass:contest}, as we discussed in \secref{sec:exes}, the G\^ateaux derivative of \(C\) is represented by \(c_F\left(x\right)=\Upsilon'\left(z_F\right)\gamma\left(x\right)\). Moreover, since gross prizes lie in \(\left[0,1\right]\), we may take \(\bar\pi=1\) in this specialization. Then \Cref{ass:contest} implies \Cref{ass:cost}\ref{ass3}, because
\[c_F\left(1\right)-c_F\left(0\right) = \Upsilon'\left(z_F\right)\gamma\left(1\right) \ge \Upsilon'\left(0\right)\gamma\left(1\right) > 1.\]
We define the rank-order benefit function
\[\Psi\left( q;\mathbf{v} \right) \coloneqq \sum_{k=1}^n \binom{n-1}{k-1} q^{n-k}\left( 1-q \right)^{k-1} v_k, \qquad \forall \ q \in \left[0,1\right].\]
For two prize vectors \(\mathbf{v}\) and \(\mathbf{w}\), we say that \(\mathbf{w}\) \textit{majorizes} \(\mathbf{v}\), if
\[\sum_{k=1}^m w_k \ge \sum_{k=1}^m v_k, \qquad \forall \ m=1,\ldots,n, \quad \text{with equality at } m = n.\]
We prove the following results in \Cref{ap:roproofs}.
\begin{lemma}\label{lem:roc}
Posit \Cref{ass:contest} and take two prize vectors \(\mathbf{v}\) and \(\mathbf{w}\) that satisfy i. \(\mathbf{w}\) majorizes \(\mathbf{v}\), and ii. \(v_1>v_2\) and \(w_1 > w_2\). Then the symmetric equilibria under the two prize vectors are unique.

Moreover, let \(F\) and \(G\) denote the (symmetric) equilibrium cdfs under prizes \(\mathbf{v}\) and \(\mathbf{w}\), and let \(X \sim F\) and \(Y \sim G\). Then, for both \(\mathbf{v}\) and \(\mathbf{w}\), the symmetric equilibrium is atomless on \(\left[0,1\right]\), and satisfies \(0 \in \operatorname{supp}\left( dF \right)\), \(0 \in \operatorname{supp}\left( dG \right)\) and \(dF\left( \left\{0\right\} \right) = dG\left( \left\{0\right\} \right)=0\), and the first-order conditions have multipliers \(\lambda_F = \lambda_G =0\).
\end{lemma}

\begin{theorem}\label{thm:icx_rank_order}
Posit \Cref{ass:contest}. Then, under the same conditions as \Cref{lem:roc}, \(Y\) dominates \(X\) in the increasing convex order.
\end{theorem}

My exercise takes a complementary reduced-form route to \citet{KimKrishnaRyvkin2023}, who explore contests in which agents ``choose their own luck:'' in an otherwise standard rank-order contest, their agents' realized performance is generated by adding arbitrary unbiased noise to effort. Their main result is striking: across a broad class of effort-cost functions \(c\), the winner-take-all prize schedule maximizes equilibrium expected output (and expected effort) relative to any other rank-order prize vector with the same budget.\footnote{Moreover, for large regions of their cost class they obtain a stronger distributional statement: when \(c\) is concave, convex, or convex-concave, the equilibrium output under winner-take-all dominates the equilibrium output under any more equal prize schedule in the increasing convex order.}

Their key insight is a two-step reduction that elegantly separates what is standard in contests from their novel ingredient produced by allowing for strategic risk taking. Using persuasion tools they derive an endogenous \emph{virtual cost of output}. This pins down equilibrium output via the familiar all-pay indifference condition and allows them to distill the outcome of moving from prizes \(\mathbf{v}\) to \(\mathbf{w}\) into (i) a \emph{prize effect}, holding the net virtual cost fixed and changing only the (capped) inverse of \(\Psi\), and (ii) a \emph{virtual cost effect}, holding \(\Psi\) fixed and changing only the net virtual cost.

The first force is exactly that emphasized in classic all-pay comparative statics (more unequal prizes increase dispersion and, with concave costs, raise output), while the second bracket isolates the novel adjustment: changes in prizes feed back into the endogenous risk-effort mix and, therefore, into the induced virtual cost. My paper takes a complementary reduced-form route in which effort and risk-taking are intertwined. My assumptions (which neither generalize nor are generalized by \citet{KimKrishnaRyvkin2023}'s) pin down the equilibrium scale of marginal costs across prize vectors, so that the classic prize-majorization force governs the comparative static.

It is also worthwhile to compare my finding with that in \citet{FangNoeStrack2020}, in which each contestant chooses an effort level \(x\) and pays a standard effort cost \(c(x)\). If a contestant mixes according to a distribution \(F\), then its total cost is simply the expected effort cost, \(C\left(F\right)=\int c dF\), which is \emph{linear in the distribution} \(F\), and whose G\^{a}teaux derivative, therefore, is independent of the distribution: \(c_{F}(x)=c(x)\) for every \(F\).\footnote{Thus, \citet{FangNoeStrack2020}'s environment fits within the present one as the special case in which marginal costs do not depend on the equilibrium distribution.}

Both this work and \citet{FangNoeStrack2020} exploit a common structural implication of majorization changes in the prize vector. Namely, that a move to a more unequal prize schedule induces a mean-preserving spread in the equilibrium distribution of marginal costs. The effect on equilibrium output is then governed by the curvature of the mapping from marginal costs back to performance. Here, the nonlinear cost enters through \(\Upsilon'\left(z_F\right)\). In equilibrium, this scalar is pinned down by the budget equation \(z_F\Upsilon'\left(z_F\right)=1/n\), and so it is \textit{common across prize vectors with the same total budget}. Accordingly, letting \(z_n^*\) denote the solution to this budget equation, more unequal prizes induce a mean-preserving spread in the equilibrium distribution of \(\Upsilon'\left(z_n^\ast\right)\gamma\left(X\right)\), and the concavity of \(\gamma\) converts this into the increasing-convex-order comparison in output.

\section{Specialization II: Risky R\&D}\label{sec:rd}

In a winner-take-all patent race, each firm chooses a distribution over breakthrough times \(T_i \in \left[0,\infty\right]\), where \(T_i=\infty\) represents never achieving a breakthrough. The firm that realizes the minimum breakthrough time obtains a unit prize discounted at (common) rate \(r>0\).

To embed this minimization problem in the abstract environment of \secref{sec:env} (where higher performance is better), we define the (decreasing) homeomorphism
\[x\left(t\right) \coloneqq \frac{1}{1+t}, \quad \forall \ t \in \left[0,\infty\right],\]
with inverse \(t\left(x\right)\coloneqq \frac{1-x}{x}\) and the convention \(t\left(0\right)\coloneqq \infty\). Since \(x\left(\cdot\right)\) is strictly decreasing, minimizing \(t\) is equivalent to maximizing \(x\). We state results directly in the time coordinate \(t\) and define \(V\left(t\right)\coloneqq e^{-rt}\), with \(V\left(\infty\right)\coloneqq 0\).

In this section, \(\mathcal F\) denotes the set of cdfs on \(\left[0,\infty\right]\). And so \(H\in\mathcal F\) is right-continuous and nondecreasing on \(\left[0,\infty\right)\), with \(H\left(\infty\right)=1\). Mass at \(\infty\) represents a positive probability of never discovering.

We specialize to the following cost form. Let \(\eta\) be a finite non-atomic positive Borel measure on \(\left[0,\infty\right]\), with \(\eta\left(\left\{\infty\right\}\right)=0\).

\begin{assumption}\label{ass:rd_monotone_marginal_cost}
For every \(H\in\mathcal F\), there exists \(m_H\in L^1\left(\eta\right)\) such that, for every \(K\in\mathcal F\),
\[\delta C\left(H;K-H\right) = \int_{\left[0,\infty\right]}m_H\left(t\right)\left(K\left(t\right)-H\left(t\right)\right)d\eta\left(t\right).\]
Moreover, for all \(F,G\in\mathcal F\),\footnote{For a function \(\varphi(x)\), we denote \(\varphi^+(x) \coloneqq \max\left\{\varphi(x),0\right\}\).}
\[\int_{\left[0,\infty\right]}\left(G\left(t\right)-F\left(t\right)\right)^+\left(m_G\left(t\right)-m_F\left(t\right)\right)d\eta\left(t\right)\ge0.\]
\end{assumption}
As a consequence, we may normalize the marginal-cost to get
\[c_H\left(\tau\right)\coloneqq \int_{\left[\tau,\infty\right]}m_H\left(t\right)d\eta\left(t\right), \quad \forall \ \tau<\infty, \qquad\text{and}\qquad c_H\left(\infty\right)\coloneqq 0.\]
This assumption is the exact-cost analog of a monotone marginal-cost condition. The object \(m_H\left(t\right)\) is the shadow cost of raising the probability of discovery by deadline \(t\). Since placing probability mass at breakthrough time \(\tau\) raises the cdf at every later finite deadline, the marginal cost \(c_H\left(\tau\right)\) is the tail integral of these deadline shadow costs. For example, if
\[C\left(H\right)=\int_{\left[0,\infty\right]}\Lambda\left(t,H\left(t\right)\right)d\eta\left(t\right),\]
where \(q\mapsto \Lambda\left(t,q\right)\) is convex and continuously differentiable, and where \(\Lambda\left(t,q\right)\) and \(\Lambda_q\left(t,q\right)\) are measurable in \(t\) and bounded in absolute value by \(\eta\)-integrable functions, then \(m_H\left(t\right)=\Lambda_q\left(t,H\left(t\right)\right)\), and the monotonicity condition in \Cref{ass:rd_monotone_marginal_cost} is immediate.

We also impose the following endpoint condition, which rules out atoms at immediate (time \(0\)) discovery:
\begin{assumption}\label{ass:rd_endpoint}
There exists \(\eta_0>0\) such that, for every \(H\in\mathcal F\), \(c_H\left(0\right)\ge1+\eta_0\).
\end{assumption}

Ties are broken uniformly. Escaping a tie by moving slightly earlier yields a discrete gain on any tie event at \(t>0\), so \Cref{ass:strict-tie} holds away from the endpoint. Under the reparameterization \(x\left(t\right)=1/\left(1+t\right)\), \Cref{lem:no-interior-atoms} rules out atoms on \(\left(0,\infty\right)\). \Cref{ass:rd_endpoint} rules out an atom at \(t=0\): moving an \(\varepsilon\)-slice of such an atom to \(\infty\) can reduce the gross prize by at most \(\varepsilon\), while convexity and \Cref{ass:rd_endpoint} imply that it lowers cost by at least \(\varepsilon\left(1+\eta_0\right)\). These combine to rule out any symmetric equilibria with atoms on \(\left[0,\infty\right)\).

For \(H\in\mathcal F\), let \(T_2,\dots,T_n\sim dH\) be i.i.d. Define \(M_H\coloneqq \min\left\{T_2,\dots,T_n\right\}\) and
\[\widetilde{V}_H\left(t\right)\coloneqq \mathbb{E}\left[V\left(M_H\right)\mathbf{1}\left\{M_H>t\right\}\right], \quad \forall \ t<\infty,\]
with \(\widetilde{V}_H\left(\infty\right)\coloneqq0\). A central object is the reduced-form private benefit
\[B_H\left(t\right)\coloneqq V\left(t\right)\left(1-H\left(t\right)\right)^{n-1}, \qquad \forall \ t<\infty,\]
with \(B_H\left(\infty\right)\coloneqq0\). If \(F\) is a symmetric equilibrium, then \(dF\left(\left\{t\right\}\right)=0\) for all \(t<\infty\), so for every \(t<\infty\) a deviator who chooses the deterministic time \(t\) wins if and only if \(M_F>t\). Thus, its interim expected prize is \(B_F\left(t\right)\).

For the (symmetric) planner, let
\[A_H(t)\coloneqq \mathbb E\left[V\left(\min\{t,M_H\}\right)\right], \qquad \forall \ t\in[0,\infty],\]
denote the gross social value when one firm is assigned deterministic breakthrough time \(t\), while the other \(n-1\) firms use \(H\). Since \(A_H(\infty)=\mathbb E[V(M_H)]\) is independent of \(t\), the planner's first-order condition can be written using the normalized marginal benefit
\[P_H(t)\coloneqq A_H(t)-A_H(\infty) = V(t)(1-H(t))^{n-1}-\widetilde V_H(t), \qquad \forall \ t<\infty,\]
with \(P_H(\infty)\coloneqq 0\). Equivalently,
\[P_H(t)=\int_t^\infty rV(s)(1-H(s))^{n-1} ds, \qquad \forall \ t<\infty.\]
The key wedge is that private firms value winning the full patent prize, whereas the planner values only the incremental acceleration of discovery relative to what rival firms would have achieved anyway. Accordingly, the symmetric planner benchmark internalizes duplication and business-stealing across firms and asks how the privately chosen breakthrough-time distribution compares with the common breakthrough-time distribution chosen by a planner who assigns the same distribution to all firms.

The planner's objective is strictly concave in this application: directly,
\[v(H) =\frac{1}{n}\mathbb E\left[V\left(\min_{1\le j\le n}T_j\right)\right]-C(H) = \frac{1}{n}\int_0^\infty rV(t)\left[1-\left(1-H(t)\right)^n\right]dt-C(H).\]
The function \(q\mapsto1-(1-q)^n\) is strictly concave on \([0,1]\). Distinct cdfs differ at a finite time and, by right-continuity, on an interval of positive length. Since \(rV(t)>0\) for every finite \(t\), the integral term is strictly concave in \(H\). And since \(C\) is convex, \(v\) is strictly concave. Consequently, the symmetric planner's optimum is unique, and the conditions in \Cref{lem:planner-kkt} are necessary and sufficient in this application.

We prove this section's result in \Cref{rdproofs}.

\begin{theorem}\label{cor:rd_overinvestment}
Posit \Cref{ass:prize,ass:cost,ass:rd_monotone_marginal_cost,ass:rd_endpoint}. Let \(F\) be a symmetric equilibrium and let \(G\) denote the unique maximizer of \(v\). Then \(F\left(t\right)\ge G\left(t\right)\) for all \(t \in \left[0,\infty\right]\).
\end{theorem}

This section's model sits in the classic patent‐race tradition, where firms compete for a priority prize and can shape not just the expected time to discovery but the entire distribution of breakthrough outcomes. \citet{dasgupta1980uncertainty} emphasize how market structure affects research speed, the number of parallel labs, and the risk profile of R\&D, while \citet{bhattacharya1986portfolio} and \citet{klette1986market} highlight that priority rewards can distort portfolio risk and duplication incentives relative to the social optimum. All investigate what competition does to the timing--\textit{viz.}, distribution--of breakthroughs.

My contribution is to abstract away from taking a particular stochastic search technology (hazards, lab counts, or a specific portfolio structure) and then deriving the implied invention‐time distribution, and instead to let each firm flexibly choose an arbitrary distribution over breakthrough time subject to an abstract cost. The payoff from doing this is a clean, stark result: equilibrium breakthrough times are stochastically earlier than the symmetric planner's, so competition accelerates discovery relative to that benchmark.

\section{Specialization III: Price and Quality Competition}\label{sec:unit_demand_prices}

We now endogenize product design in \citet{perloff1985equilibrium}'s seminal model of oligopolistic price competition. Our players remain firms. Each chooses a distribution \(F_i\in\mathcal F\) over product quality \(Q_i\in X\coloneqq[0,1]\) and a price \(p_i\in P\coloneqq[0,\bar p]\). We normalize marginal cost to \(0\), so \(p_i\) is the per-unit markup. To justify a first-order approach in prices, we impose a taste shock, which smooths demand.
\begin{assumption}\label{as:price_smoothing}
Take some \(\sigma\in(0,1)\). After a firm draws \(Q_i\sim dF_i\), the representative consumer draws an independent \(\varepsilon_i\) with cdf \(G\) on \([0,1]\) that admits a continuous density \(g\). The consumer's value for product \(i\) is \(\hat{Q}_i \coloneqq (1-\sigma)Q_i+\sigma \varepsilon_i\in[0,1]\).
\end{assumption}
We assume full market coverage. To wit, the consumer ranks firms by \(\hat{Q}_i-p_i\) and purchases from the firm maximizing this, with ties broken uniformly. Firm \(i\)'s gross operating profit is \(p_i\) if it makes the sale and \(0\) otherwise, and it incurs the distributional cost \(C(F_i)\).

Under \Cref{as:price_smoothing}, for any \(F\in\mathcal F\) the induced cdf of \(\hat{Q}\) is absolutely continuous, so ties occur with probability \(0\) for every profile. For any \(F\in\mathcal F\), write \(\hat{F}\) for the induced cdf of \(\hat{Q}=(1-\sigma)Q+\sigma\varepsilon\) when \(Q\sim dF\) and \(\varepsilon\sim G\) are independent, and write \(\hat{f}\) for its density.

Fix symmetric opponents \(\left(F,p\right)\in\mathcal F\times P\), and let \(\hat{F}\) denote the induced cdf of opponents' effective quality \(\hat{Q}\). Consider a deviator that chooses \(\left(H,r\right)\in\mathcal F\times P\), draws \(Q\sim dH\), and posts price \(r\). Conditional on realizing quality \(q\) and taste shock \(\varepsilon\), the deviator wins if and only if
\[(1-\sigma)q+\sigma\varepsilon-r > \max_{2\le j\le n}\left\{\hat{Q}_j-p\right\},\]
if and only if \(\hat{Q}_j < (1-\sigma)q+\sigma\varepsilon+p-r\) for all \(j\ge 2\). Accordingly, writing \(\hat{F}(x)\coloneqq 0\) for \(x<0\) and \(\hat{F}(x)\coloneqq 1\) for \(x>1\), the deviator's interim expected profit is
\[a_{F,p}\left(q;r\right) \coloneqq r \cdot \int_0^1 \left(\hat{F}\left((1-\sigma)q+\sigma e+p-r\right)\right)^{n-1} dG(e).\]
For each \(r\in P\), denote
\[W_{F,p}(r) \coloneqq \max_{H\in\mathcal F} \left\{ \int_X a_{F,p}\left(q;r\right) dH(q)-C(H) \right\}.\]
Note that \(W_{F,p}(r)\) already incorporates the best distributional deviation at price \(r\), so optimizing over \(r\) rules out \emph{double deviations}.

\begin{proposition}\label{prop:price_kkt_and_price_foc}
Posit \Cref{as:price_smoothing,ass:cost}. Let \(\left(F,p\right)\in\mathcal F\times P\) be a symmetric equilibrium with an interior \(p\),\footnote{In some sense, the specialization to an equilibrium with a deterministic choice of price and distribution is not innocuous--one may not exist. However, for the purposes of the main exercise of this section (providing implications of marginal-cost pricing in the large-market limit) this restriction is harmless and merely for expositional ease. To wit, in the Supplementary Appendix, I show that analogous results hold in symmetric equilibria where firms jointly randomize over prices and quality distributions.} and let \(\hat{F}\) denote the induced cdf of \(\hat{Q}=(1-\sigma)Q+\sigma\varepsilon\) under \(F\), with density \(\hat{f}\). Then,
\begin{enumerate}[noitemsep]
\item\label{priceitem1} There exists \(\lambda_g\in\mathbb R\) such that
\[\begin{split} p \int_0^1 \left(\hat{F}\left((1-\sigma)q+\sigma e\right)\right)^{n-1} dG(e)-c_F(q) &\le \lambda_g\quad \forall \ q\in X, \quad \text{and}\\ p \int_0^1 \left(\hat{F}\left((1-\sigma)q+\sigma e\right)\right)^{n-1} dG(e)-c_F(q) &= \lambda_g \quad \forall \ q\in\supp(dF); \quad \text{and} \end{split}\]
\item\label{priceitem2} The equilibrium price \(p = \frac{1}{n \mathbb E\left[\hat{f}\left(\bar{Q}\right)\right]}\), where \(\bar{Q} \coloneqq \max\left\{\hat{Q}_2,\dots,\hat{Q}_n\right\}\).
\end{enumerate}
\end{proposition}
\begin{proof}
\eqref{priceitem1} Fix opponents at \(\left(F,p\right)\). At price \(p\), the deviator optimizes
\[\int_X \left[p\int_0^1 \left(\hat{F}\left((1-\sigma)q+\sigma e\right)\right)^{n-1} dG(e)\right] dH(q)-C(H)\]
over \(H\in\mathcal F\). \(F\) is a best response, so \Cref{lem:game-kkt} produces the stated conditions.

\eqref{priceitem2} Under \Cref{as:price_smoothing}, ties occur with probability \(0\), and \(\hat{F}\) is absolutely continuous with density \(\hat{f}\). Since \(g\) is continuous on \([0,1]\), it is bounded \(\Rightarrow\) so too is \(\hat{f}\). Consider the \emph{price-only} deviation that holds the distribution fixed at \(F\) (so \(\hat{F}\) is also fixed) and posts price \(r\in P\). Its expected gross profit is
\[\iota(r)\coloneqq \int_X a_{F,p}\left(q;r\right) dF(q) = \int_X r\cdot \left(\hat{F}\left(z+p-r\right)\right)^{n-1} d\hat{F}(z).\]
As \(\left(F,p\right)\) is a best response to opponents' \(\left(F,p\right)\), \(p\) maximizes \(\iota\) over \(P\). Since \(p\) is interior, \(\iota'(p)=0\). Moreover, writing \(\hat{f}(x)\coloneqq 0\) for \(x\notin[0,1]\), we have
\[\partial_r\left[r\cdot \left(\hat{F}\left(z+p-r\right)\right)^{n-1}\right] = \left(\hat{F}\left(z+p-r\right)\right)^{n-1} - r\left(n-1\right)\left(\hat{F}\left(z+p-r\right)\right)^{n-2} \hat{f}\left(z+p-r\right),\]
for \(d\hat{F}\)-a.e. \(z\). Boundedness permits us to differentiate under the integral sign, so \(\iota'(r)=\int_X \partial_r\left[r\cdot \left(\hat{F}\left(z+p-r\right)\right)^{n-1}\right] d\hat{F}(z)\). Evaluating at \(r=p\) yields
\[0 = \iota'(p) = \int_X \left[ \left(\hat{F}(z)\right)^{n-1} - p\left(n-1\right)\left(\hat{F}(z)\right)^{n-2} \hat{f}(z) \right]d\hat{F}(z).\]
By symmetry, \(\frac{1}{n}= \int_X \left(\hat{F}(z)\right)^{n-1} d\hat{F}(z)\). Finally, \(\bar{Q}\) has density \(\left(n-1\right)\left(\hat{F}(z)\right)^{n-2}\hat{f}(z)\), so
\[\mathbb E\left[\hat{f}\left(\bar{Q}\right)\right] = \int_0^1 \hat{f}(z)\left(n-1\right)\left(\hat{F}(z)\right)^{n-2}\hat{f}(z) dz = \left(n-1\right)\int_0^1 \left(\hat{F}(z)\right)^{n-2} \hat{f}(z)^2 dz,\]
which produces \(p = 1/\left(n \mathbb E\left[\hat{f}\left(\bar{Q}\right)\right]\right)\).
\end{proof}

One notable implication of this result is that the equilibrium marginal-cost gap is pinned down by the equilibrium price, scaled by the win-probability gap between the endpoints of the equilibrium quality support:

\begin{corollary}\label{cor:price_equals_mc_gap}
Posit \Cref{ass:cost,as:price_smoothing}. Let \(\left(F,p\right)\) be a symmetric equilibrium with \(p>0\), and let \(\hat{F}\) denote the induced cdf of \(\hat{Q}=(1-\sigma)Q+\sigma\varepsilon\) under \(F\). Define the support endpoints \(\ubar{q} \coloneqq \inf\supp(dF)\) and \(\bar{q} \coloneqq \sup\supp(dF)\). Define, for each \(q\in X\),
\[\omega(q)\coloneqq \int_0^1 \left(\hat{F}\left((1-\sigma)q+\sigma e\right)\right)^{n-1} dG(e).\]
Then, \(c_F\left(\bar{q}\right)-c_F\left(\ubar{q}\right) = p\cdot\left(\omega(\bar{q})-\omega(\ubar{q})\right)\); and, notably, \(0\le c_F\left(\bar{q}\right)-c_F\left(\ubar{q}\right)\le p\).
\end{corollary}

\begin{proof}
By \Cref{prop:price_kkt_and_price_foc}(\ref{priceitem1}), for every \(q\in\supp(dF)\), \(p\omega(q)-c_F(q)=\lambda_g\). Since \(\ubar{q},\bar{q}\in\supp(dF)\), subtracting the equality at \(\ubar{q}\) from that at \(\bar{q}\) yields
\[c_F\left(\bar{q}\right)-c_F\left(\ubar{q}\right) = p\cdot\left(\omega(\bar{q})-\omega(\ubar{q})\right).\]
Finally, \(\hat{F}(\cdot)\in[0,1]\) \(\Longrightarrow\)  \(0\le \omega(\bar{q})-\omega(\ubar{q})\le 1\)  \(\Longrightarrow\) \(0\le c_F(\bar{q})-c_F(\ubar{q})\le p\).
\end{proof}
This simple insight allows us to deduce a necessary condition for competitive pricing in the ``perfect competition'' (many-firm) limit.

\begin{corollary}\label{cor:necessary_mc_limit}
For each \(n\ge 2\), let \(\left(F_n,p_n\right)\) be a symmetric equilibrium in the \(n\)-firm game, and define \(\ubar{q}_n \coloneqq \inf\supp\left(dF_n\right)\) and \(\bar{q}_n \coloneqq \sup\supp\left(dF_n\right)\). Then,
\begin{enumerate}[noitemsep]
\item\label{lastcorr1} If \(p_n\to 0\) as \(n\to\infty\), then \(c_{F_n}\left(\bar{q}_n\right)-c_{F_n}\left(\ubar{q}_n\right)\to 0\).
\item Moreover, if there exists \(\kappa>0\) such that for all \(F\in\mathcal F\) and all \(0\le x<y\le 1\), \(c_F(y)-c_F(x)\ge \kappa\left(y-x\right)\), then \(p_n\to 0\) implies \(\bar{q}_n-\ubar{q}_n\to 0\).
\end{enumerate}
\end{corollary}
\begin{proof}
For each \(n\), \Cref{cor:price_equals_mc_gap} implies
\[c_{F_n}\left(\bar{q}_n\right)-c_{F_n}\left(\ubar{q}_n\right) = p_n\cdot\left(\omega_n(\bar{q}_n)-\omega_n(\ubar{q}_n)\right) \le p_n,\]
since \(\omega_n(\cdot)\in[0,1]\). Thus, \(p_n\to 0\) \(\Longrightarrow\) \(c_{F_n}\left(\bar{q}_n\right)-c_{F_n}\left(\ubar{q}_n\right)\to 0\), proving the first claim. For the second, under the stated uniform-steepness property, \(c_{F_n}\left(\bar{q}_n\right)-c_{F_n}\left(\ubar{q}_n\right) \ge \kappa\left(\bar{q}_n-\ubar{q}_n\right)\), so \(c_{F_n}\left(\bar{q}_n\right)-c_{F_n}\left(\ubar{q}_n\right)\to 0\) \(\Longrightarrow\) \(\bar{q}_n-\ubar{q}_n\to 0\). Combine with the first claim.
\end{proof}

\citet{HwangKimBoleslavsky2023} also endogenize firms' valuation distributions in the \citet{perloff1985equilibrium} environment. However, in contrast to my exercise, their distribution choice corresponds to \textit{advertising}; \textit{viz.}, information-provision. Their large-market message is that when there are many firms, winning a customer becomes hard unless a firm looks exceptionally good. That competitive pressure pushes firms toward revealing more: for sufficiently large \(n\), a fully revealing (full-information) outcome is an equilibrium of their advertising game, and in a number of extensions equilibrium advertising becomes arbitrarily close to full revelation as \(n \to \infty\).

My exercise has the same broad structure--firms choose price and a second margin that affects how likely they are to be chosen--but the second margin is the product's quality, generated at cost, rather than information. Consequently, large markets need not force prices down to marginal cost. Instead, markups evaporate only if the marginal cost of producing outcomes near the top of the equilibrium quality range becomes vanishingly close to the marginal cost near the bottom. If that cost gap does not disappear, positive markups persist even with many firms.

\bibliography{sample.bib}

\appendix

\section{Further Technical Details \& Equilibrium Existence Proof}\label[appendix]{app:technical-details}

An equivalent formulation of each player's strategy is that each is choosing a Borel probability measure \(\mu\equiv dF\) on \(X\) satisfying \(F(x)\coloneqq \mu\left(\left[0,x\right]\right)\) for all \(x \in \left[0,1\right]\), with \(F^-(1)\coloneqq \lim_{x\uparrow 1}F(x)=1-\mu(\{1\})\). Conversely, every \(F\in\mathcal F\) induces a unique probability measure \(dF\) on \(X\). We write \(\supp(dF)\subseteq X\) for its support. A sequence \(F_m\to F\) in \(\mathcal F\) means weak convergence of the induced measures \(dF_m\Rightarrow dF\) on \(X\).

We use the following notation and conventions in the appendix. For a set \(A\) in a vector space, \(\co(A)\) denotes its convex hull.\footnote{In particular, for \(F^1,\dots,F^m\in\mathcal F\), \(\co\{F^1,\dots,F^m\} = \left\{\sum_{k=1}^m \lambda_k F^k \colon \lambda_k\geq 0,\ \sum_{k=1}^m \lambda_k=1\right\}\).} For a function \(f\) between topological spaces, \(\mathrm{Disc}(f)\) denotes the set of points at which \(f\) is not continuous. \(\delta_t\) denotes the Dirac measure at \(t\in X\). The symbol \(\otimes\) indicates the product measure.\footnote{If \(\mu\) is a measure on \(X\) and \(\nu\) is a measure on \(Y\), then \(\mu \otimes \nu\) is the product measure on \(X \times Y\), satisfying \(\left(\mu\otimes\nu\right)(A\times B)=\mu(A)\nu(B)\) on rectangles.} Integrals against measures are Lebesgue integrals with respect to the corresponding probability measures--in one dimension this is the usual Lebesgue-Stieltjes notation.

The expectation notation in the main text is shorthand for integration against the corresponding product distributions. For example, if \(X_i\sim dF_i\) independently, then
\[\mathbb E\left[\pi_i\left(X_1,\ldots,X_n\right)\right] = \int_{X^n}\pi_i\left(x_1,\ldots,x_n\right)d\left(dF_1\otimes\cdots\otimes dF_n\right)\left(x\right).\]

Uniform tie-breaking can be implemented explicitly as follows. Let \(R_i\sim \mathrm{Unif}\left[0,1\right]\) be i.i.d. and independent of \(\left(X_1,\ldots,X_n\right)\). If \(\widehat{\pi}_i\colon\left(X\times\left[0,1\right]\right)^n\to\R\) denotes the gross payoff before averaging over tie-breakers, then the main-text payoff \(\pi_i\colon X^n\to\R\) is
\[\pi_i\left(x_1,\ldots,x_n\right) \coloneqq \int_{\left[0,1\right]^n} \widehat{\pi}_i\left(\left(x_1,r_1\right),\ldots,\left(x_n,r_n\right)\right)d\lambda^{\otimes n}\left(r\right).\]
Equivalently, for a profile \((F_1,\ldots,F_n)\),
\[\mathbb E\left[\pi_i\left(X_1,\ldots,X_n\right)\right] = \int_{\left(X\times\left[0,1\right]\right)^n} \widehat{\pi}_i d\left(\left(dF_1\otimes\lambda\right)\otimes\cdots\otimes\left(dF_n\otimes\lambda\right)\right).\]
In short, the expectation notation in the main text is equivalent to this explicit product-measure formulation with independent uniform tie-breakers.

\subsection{Proof of \texorpdfstring{\Cref{prop:symNE}, Equilibrium Existence}{Proposition~\ref{prop:symNE}, Equilibrium Existence}}\label[appendix]{existenceproof}

We prove \Cref{prop:symNE} through a sequence of lemmas verifying that the assumptions for \citet[Theorem 4.1]{reny1999existence} hold, which guarantees the existence of a symmetric pure-strategy Nash equilibrium. For \(F\in\mathcal{F}\), define the diagonal payoff \(v(F)\coloneqq u_1(F,\dots,F)\).

\begin{lemma}\label{lem:vcont}
Under \Cref{ass:prize,ass:cost}, \(v\colon\mathcal{F}\to\R\) is continuous.
\end{lemma}
\begin{proof}[Proof of \Cref{lem:vcont}]
Under symmetry, \(\mathbb E[\pi_i]\) is the same for all \(i\). Summing over \(i\) and using \Cref{ass:prize}\ref{ass:prize-agg} yields
\[v(F) = \frac{1}{n}\mathbb E\left[\Pi(X_1,\ldots,X_n)\right]-C(F) = \frac{1}{n}\int_{X^n}\Pi(x)d(dF)^{\otimes n}(x)-C(F).\]
The integrand is bounded and continuous on the compact \(X^n\). \(F\mapsto (dF)^{\otimes n}\) is continuous under weak convergence, and \(C\) is continuous by \Cref{ass:cost}\ref{ass1}.
\end{proof}

Our next lemma establishes diagonal quasiconcavity.
\begin{lemma}\label{lem:dqc}
Under \Cref{ass:cost}, for any \(m\ge 2\), any \(F^1,\dots,F^m\in\mathcal{F}\), and any \(F\in\co\{F^1,\dots,F^m\}\), \(v(F)\geq \min_{1\le k\le m} u_1(F^k,F,\dots,F)\).
\end{lemma}

\begin{proof}[Proof of \Cref{lem:dqc}]
Fix \(F\in\mathcal{F}\). The mapping \(H\mapsto u_1(H,F,\dots,F)\) is concave in \(H\): the expected prize term is affine in \(dH\), and \(-C(H)\) is concave because \(C\) is convex. If \(F=\sum_{k=1}^m \lambda_k F^k\) with \(\lambda_k\ge 0\) and \(\sum_k\lambda_k=1\), then concavity produces
\[v(F)=u_1(F,F,\dots,F)\ge \sum_{k=1}^m \lambda_k u_1(F^k,F,\dots,F)\ge \min_{1\le k\le m} u_1(F^k,F,\dots,F). \qedhere\]
\end{proof}

We verify diagonal better-reply security via an atomless smoothing device, generated over several lemmas.

\begin{lemma}\label{lem:atomless_cont}
Posit \Cref{ass:prize}. Fix \(\hat H\in\mathcal{F}\) such that \(d\hat H(\{x\})=0\) for every \(x\in X\). Then the map \(F\mapsto u_1(\hat H,F,\dots,F)\) is continuous on \(\mathcal{F}\).
\end{lemma}

\begin{proof}[Proof of \Cref{lem:atomless_cont}]
Let \(F_m\to F\). Then \(d\hat H\otimes(dF_m)^{\otimes(n-1)} \to d\hat H\otimes(dF)^{\otimes(n-1)}\) weakly on \(X^n\). Consider the bounded function \(\pi_1\) on \(X^n\). By \Cref{ass:prize}\ref{ass:prize-disc}, \(\pi_1\) is continuous at every \(x\in X^n\) such that \(x_1\ne x_j\) for all \(j\ge 2\). Hence, \(\mathrm{Disc}(\pi_1) \subseteq \bigcup_{j=2}^n \left\{x\in X^n: x_1=x_j\right\}\). Under the limiting product measure \(d\hat H\otimes(dF)^{\otimes(n-1)}\), for each \(j\ge 2\),
\[\mathbb P(x_1=x_j)=\int_X d\hat H(\{x\})dF(x)=0,\]
since \(d\hat H\) is atomless on \(X\). Thus, the limiting measure assigns zero mass to \(\mathrm{Disc}(\pi_1)\).

By the extended Portmanteau theorem: if \(M_m\to M\) weakly and \(f\) is bounded with \(M(\mathrm{Disc}(f))=0\), then \(\int f dM_m\to\int f dM\). Applying this to \(f=\pi_1\) yields convergence of the expected prize term. Since \(C(\hat H)\) does not depend on \(F\), continuity follows.
\end{proof}

Next we show that no best response places positive mass at \(1\).

\begin{lemma}\label{lem:one_dominated}
Posit \Cref{ass:prize,ass:cost}. Fix opponents \((F_2,\dots,F_n)\in\mathcal{F}^{n-1}\) and a strategy \(H\in\mathcal{F}\). Let \(\alpha\coloneqq dH(\{1\})\), and define \(\widetilde H\in\mathcal{F}\) by moving this mass to \(0\): \(d\widetilde H \coloneqq dH+\alpha(\delta_0-\delta_1)\). Then,
\[u_1(\widetilde H,F_2,\dots,F_n)\ge u_1(H,F_2,\dots,F_n)+\eta_1 \alpha.\]
\end{lemma}

\begin{proof}[Proof of \Cref{lem:one_dominated}]
By \Cref{ass:prize}\ref{ass:prize-bdd}, \(0\le \pi_1\le \bar\pi\), hence, moving probability \(\alpha\) from \(1\) to \(0\) can reduce the expected prize term by at most \(\bar{\pi} \alpha\).

For the cost, by convexity of \(C\), \(C(H)-C(\widetilde H) \ge \delta C\left(\widetilde H;H-\widetilde H\right)\). Since \(d\widetilde H=dH+\alpha(\delta_0-\delta_1)\), we have \(d\left(H-\widetilde H\right)=\alpha(\delta_1-\delta_0)\). Thus,
\[\delta C\left(\widetilde H;H-\widetilde H\right) = \int_X c_{\widetilde H}(x)d\left(H-\widetilde H\right)(x) = \alpha\left(c_{\widetilde H}(1)-c_{\widetilde H}(0)\right) \ge \alpha\left(\bar\pi+\eta_1\right).\]
with the first equality and the inequality granted by \Cref{ass:cost}\ref{ass2} and \ref{ass3}. Therefore,
\[C(\widetilde H)-C(H)\le -\alpha\left(\bar\pi+\eta_1\right),\]
and so the net payoff gain is at least \(-\bar{\pi} \alpha+\alpha(\bar\pi+\eta_1)=\eta_1 \alpha\).\end{proof}

Now we smooth profitable deviations.
\begin{lemma}\label{lem:smooth_profitable}
Posit \Cref{ass:prize,ass:cost}. Fix \(F\in\mathcal{F}\). If there exist \(H^0\in\mathcal{F}\) and \(\eta>0\) such that
\[u_1(H^0,F,\dots,F)\ge v(F)+\eta,\]
then there exists \(\hat H\in\mathcal{F}\) with \(d\hat H(\{x\})=0\) for every \(x\in X\) such that
\[u_1(\hat H,F,\dots,F)\ge v(F)+\frac{\eta}{2}.\]
\end{lemma}

\begin{proof}[Proof of \Cref{lem:smooth_profitable}]
First we remove any mass at \(1\). Let \(\alpha\coloneqq dH^0(\{1\})\) and let \(\widetilde H^0\) be obtained from \(H^0\) by moving \(\alpha\) to \(0\), as in \Cref{lem:one_dominated}. We have \(u_1(\widetilde H^0,F,\dots,F)\ge u_1(H^0,F,\dots,F)\ge v(F)+\eta\). In particular, \(d\widetilde H^0(\{1\})=0\).

Next let \(X_1\sim d\widetilde H^0\). For \(\delta>0\), define a smoothed \(\hat X_1^\delta\) by: if \(X_1=1\), set \(\hat X_1^\delta=1\); and if \(X_1=x<1\), draw an independent \(V\sim\mathrm{Unif}\left[0,1\right]\) and set
\[\hat X_1^\delta \coloneqq x + V \min\left\{\delta,1-x\right\}.\]
Let \(\hat H^\delta\) be the cdf induced by \(\hat X_1^\delta\). Then \(d\hat H^\delta(\{x\})=0\) for every \(x\in X\), and \(\hat H^\delta\to \widetilde H^0\) weakly as \(\delta\downarrow 0\).

Fix opponents at \(F\). Conditional on opponents' realizations, increasing \(x_1\) cannot decrease firm \(1\)'s prize payoff by \Cref{ass:prize}\ref{ass:prize-mono}. Since \(\hat X_1^\delta\ge X_1\) a.s., the expected prize term under \(\hat H^\delta\) against \(F\) is at least that under \(\widetilde H^0\).

By \Cref{ass:cost}\ref{ass1} and \(\hat H^\delta\to \widetilde H^0\), we have \(C(\hat H^\delta)\to C(\widetilde H^0)\). Choose \(\delta\) such that \(\left|C(\hat H^\delta)-C(\widetilde H^0)\right|\le \eta/2\). Then
\[u_1(\hat H^\delta,F,\dots,F)\ge u_1(\widetilde H^0,F,\dots,F)-\frac{\eta}{2}\ge v(F)+\frac{\eta}{2}.\]
Set \(\hat H\coloneqq \hat H^\delta\).
\end{proof}

Next, we establish \citet{reny1999existence}'s diagonal better-reply security.
\begin{lemma}\label{lem:dbrs}
Under \Cref{ass:prize,ass:cost}, the game is diagonally better-reply secure.
\end{lemma}

\begin{proof}[Proof of \Cref{lem:dbrs}]
Let \(F\in\mathcal{F}\) be such that \((F,\dots,F)\) is not a symmetric equilibrium. Then there exist \(H^0\in\mathcal{F}\) and \(\eta>0\) with \(u_1(H^0,F,\dots,F)\ge v(F)+\eta\). By \Cref{lem:smooth_profitable}, there exists \(\hat H\) atomless on \(X\) with \(u_1(\hat H,F,\dots,F)\ge v(F)+\eta/2\). By \Cref{lem:atomless_cont}, the map \(G\mapsto u_1(\hat H,G,\dots,G)\) is continuous at \(F\), so there exists a neighborhood \(U\) of \(F\) such that \(u_1(\hat H,G,\dots,G)>v(F)\) for all \(G\in U\). Thus, a player can secure strictly above \(v(F)\) along the diagonal at \(F\).
\end{proof}

\begin{proof}[Proof of \Cref{prop:symNE}]
The game is symmetric with compact metric strategy space \(\mathcal{F}\). By \Cref{lem:dqc} it is diagonally quasiconcave, and by \Cref{lem:dbrs} it is diagonally better-reply secure. \citet[Theorem 4.1]{reny1999existence} implies the existence of a symmetric equilibrium.
\end{proof}

\section{Proofs of Lemmas \ref{lem:no-atom-1}, \ref{lem:no-interior-atoms}, and \ref{lem:aF-cont}}\label[appendix]{app:b}
In this appendix, we collect the proofs of Lemmas \ref{lem:no-atom-1}, \ref{lem:no-interior-atoms} and \ref{lem:aF-cont}.

\begin{proof}[Proof of \Cref{lem:no-atom-1}]
If \(dF(\{1\})>0\), then applying \Cref{lem:one_dominated} with \(H=F\) and opponents all equal to \(F\) yields a profitable deviation, contradicting equilibrium.
\end{proof}

\begin{proof}[Proof of \Cref{lem:no-interior-atoms}]
Suppose for the sake of contradiction that \(dF(\{x_0\})=\alpha>0\) for some \(x_0\in\left(0,1\right)\). Let \(a\coloneqq F(x_0^-)\in[0,1)\), so \(F(x_0)=a+\alpha\).

Fix \(\delta\in\left(0,\delta(x_0)\right]\). Let \(\nu_\delta\) be the uniform probability measure on \(\left(x_0,x_0+\delta\right]\), and define a deviation \(H_{\varepsilon,\delta}\in\mathcal F\) by \(dH_{\varepsilon,\delta}\coloneqq dF + \varepsilon\left(\nu_\delta-\delta_{x_0}\right)\), which removes an \(\varepsilon\)-slice of the atom at \(x_0\) and redistributes it continuously over \(\left(x_0,x_0+\delta\right]\).

First we observe that the expected prize gain (gross of costs) is first-order positive. Take the event \(E\coloneqq \left\{\max_{2\le j\le n} X_j = x_0\right\}\). Since \(dF(\{x_0\})=\alpha>0\), we have
\[\mathbb P(E)=F(x_0)^{n-1}-F(x_0^-)^{n-1}=\left(a+\alpha\right)^{n-1}-a^{n-1}>0.\]
On \(E\), if player \(1\) realizes \(x_0\), it is tied for the top value \(x_0\) with at least one opponent. By \Cref{ass:strict-tie}, moving its realized value to any \(x_0+\varepsilon'\in\left(x_0,x_0+\delta\right]\) raises its expected prize by at least \(\Delta(x_0)\). By \Cref{ass:prize}\ref{ass:prize-mono}, increasing player \(1\)'s realized value never decreases its prize on any realization. Therefore, the expected prize gain from shifting the \(\varepsilon\)-slice of mass is at least \(\varepsilon\cdot \Delta(x_0)\cdot \mathbb P(E) > 0\).

Second, we argue that the cost change can be made (comparatively) small. Fix \(\delta > 0\). Define \(G_\delta\in\mathcal F\) by \(dG_\delta \coloneqq dF+\alpha(\nu_\delta-\delta_{x_0})\), so that \(H_{\varepsilon,\delta}=(1-s)F+sG_\delta\), with \(s=\varepsilon/\alpha\). By the G\^{a}teaux differentiability at \(F\) in direction \(G_\delta-F\),
\[\lim_{\varepsilon\downarrow 0}\frac{C(H_{\varepsilon,\delta})-C(F)}{\varepsilon} = \int_X c_F(x) d(\nu_\delta-\delta_{x_0})(x) = \frac{1}{\delta}\int_{x_0}^{x_0+\delta}c_F(x) dx-c_F(x_0).\]
Hence, for any \(\kappa > 0\) there exists \(\bar\varepsilon > 0\) such that for all \(0<\varepsilon<\bar\varepsilon\),
\[\left| C(H_{\varepsilon,\delta})-C(F) -\varepsilon\left(\frac{1}{\delta}\int_{x_0}^{x_0+\delta}c_F(x) dx-c_F(x_0)\right) \right| \le \varepsilon\kappa.\]
Since \(c_F\) is continuous, choose \(\delta > 0\) so small that
\[\left|\frac{1}{\delta}\int_{x_0}^{x_0+\delta}c_F(x) dx-c_F(x_0)\right| \le \Delta(x_0)\mathbb P(E)/4.\]
Then set \(\kappa\coloneqq \Delta(x_0)\mathbb P(E)/4\) and choose \(\varepsilon\in(0,\min\{\alpha,\bar\varepsilon\})\). For these choices, the cost increase is at most \(\varepsilon \Delta(x_0)\mathbb P(E)/2\).

Thus, the net increase is at least \(\varepsilon \Delta(x_0)\mathbb P(E)-\varepsilon \Delta(x_0)\mathbb P(E)/2 > 0\), a strictly profitable deviation, which is a contradiction.
\end{proof}

\begin{proof}[Proof of \Cref{lem:aF-cont}]
Let \(x_m\to x_0\). Let \(X_2,\ldots,X_n\sim dF\) be i.i.d. Then \(a_F(x_m)=\mathbb E\left[\pi_1\left(x_m,X_2,\ldots,X_n\right)\right]\) and \(a_F(x_0)=\mathbb E\left[\pi_1\left(x_0,X_2,\ldots,X_n\right)\right]\). As \(\pi_1\) is bounded, \(\left|\pi_1\left(x_m,X_2,\ldots,X_n\right)\right|\le \bar\pi\) for all \(m\). Let \(E\coloneqq \left\{X_j\ne x_0 \text{ for all } j\in\{2,\ldots,n\}\right\}\). Since \(dF(\{x_0\})=0\), we have \(\mathbb P(E)=1\).

Fix \(\omega\in E\). Then \(\left(x_0,X_2(\omega),\ldots,X_n(\omega)\right)\) satisfies \(x_0\ne X_j(\omega)\) for all \(j\ge 2\), so \(\pi_1\) is continuous there by \Cref{ass:prize}\ref{ass:prize-disc}. Thus,
\[\pi_1\left(x_m,X_2(\omega),\ldots,X_n(\omega)\right)\to \pi_1\left(x_0,X_2(\omega),\ldots,X_n(\omega)\right).\]
The dominated convergence theorem yields \(a_F(x_m)\to a_F(x_0)\).
\end{proof}

\section{First-Order Approach Proofs}\label[appendix]{app:kkt}

\subsection{Auxiliary Lemmas}

\begin{lemma}\label{lem:supp_interval}
Let \(F\in\mathcal{F}\) and \(0\le a<b \leq 1\). If \(\supp(dF)\cap(a,b]=\varnothing\), then \(dF((a,b])=0\), i.e., \(F\) is constant on \([a,b]\).
\end{lemma}

\begin{proof}
Fix \(k\ge 1\). The set \([a+1/k,b]\) is compact and contained in \((a,b]\), hence, disjoint from \(\supp(dF)\). Thus, \(dF([a+1/k,b])=0\). Since \((a,b]\subseteq (a,a+1/k]\cup[a+1/k,b]\),
\[0\le dF((a,b])\le dF((a,a+1/k])=F(a+1/k)-F(a).\]
Letting \(k\to\infty\) and using right-continuity yields \(dF((a,b])=0\).
\end{proof}

Fix a bounded Borel function \(a\colon X\to\R\). Define
\[J_a(F)\coloneqq \int_X a(t) dF(t)-C(F),\quad  \forall \ F\in\mathcal{F}.\]

\begin{lemma}\label{lem:var_ineq}
Posit \Cref{ass:cost}\ref{ass1} and \ref{ass2}. If \(F^\star\) maximizes \(J_a\) over \(\mathcal{F}\), then for every \(G\in\mathcal{F}\),
\[\int_X \left(a(t)-c_{F^\star}(t)\right) d\left(G-F^\star\right)(t)\le 0.\]
\end{lemma}

\begin{proof}
Fix \(G\in\mathcal{F}\) and set \(F_\epsilon\coloneqq (1-\epsilon)F^\star+\epsilon G\). The map \(\epsilon\mapsto \int a dF_\epsilon\) is affine, and by convexity of \(C\), \(\epsilon\mapsto -C(F_\epsilon)\) is concave. Consequently, \(\epsilon\mapsto J_a(F_\epsilon)\) is concave and maximized at \(\epsilon=0\), so the right derivative at \(0\) is nonpositive. \Cref{ass:cost}\ref{ass2} provides the derivative in the stated integral representation, yielding the inequality.
\end{proof}

\begin{lemma}\label{lem:kkt_simplex}
Posit \Cref{ass:cost}\ref{ass1} and \ref{ass2}. If \(F^\star\) maximizes \(J_a\) over \(\mathcal{F}\), then there exists \(\lambda\in\R\) such that, writing \(h(t)\coloneqq a(t)-c_{F^\star}(t)\),
\[h(t)\le \lambda\quad \forall t\in X, \quad \text{and} \quad \int_X \left(\lambda-h(t)\right) dF^\star(t)=0.\]
\end{lemma}

\begin{proof}[Proof of \Cref{lem:kkt_simplex}]
For any \(t\in X\), let \(G_t\in\mathcal{F}\) be degenerate at \(t\) (so \(dG_t=\delta_t\)). From \Cref{lem:var_ineq},
\[h(t)-\int_X h(u) dF^\star(u)=\int_X h(u) d(G_t-F^\star)(u)\le 0.\]
Let \(\lambda\coloneqq \int_X h(u) dF^\star(u)\). Then \(h\le \lambda\), and \(\int (\lambda-h) dF^\star=0\) by definition.
\end{proof}

\begin{lemma}\label{lem:compl_support}
Let \(F\in\mathcal{F}\), let \(h\colon X\to\R\) be Borel, and let \(\lambda\in\R\) satisfy \(h\le \lambda\) and \(\int (\lambda-h) dF=0\). Then:
\begin{enumerate}[noitemsep,label=(\roman*),leftmargin=*]
\item\label{it:b41} If \(dF(\{t_0\})>0\), then \(h(t_0)=\lambda\); and
\item\label{it:b42} If \(t_0\in\supp(dF)\) and \(h\) is continuous at \(t_0\), then \(h(t_0)=\lambda\).
\end{enumerate}
\end{lemma}

\begin{proof}\ref{it:b41} If \(h(t_0)<\lambda\) and \(dF(\{t_0\})>0\), then \(\int (\lambda-h) dF\ge (\lambda-h(t_0))dF(\{t_0\})>0\), a contradiction.

\ref{it:b42} If \(h(t_0)<\lambda\) and \(h\) is continuous at \(t_0\), then there exist \(\epsilon,\delta>0\) such that \(h\le \lambda-\epsilon\) on \(U\coloneqq (t_0-\delta,t_0+\delta)\cap X\). As \(t_0\in\supp(dF)\), \(dF(U)>0\), hence, \(\int (\lambda-h) dF\ge \epsilon dF(U)>0\), a contradiction.
\end{proof}

Define the planner prize functional
\[\tag{\(C1\)}\label{eq:a1}\widetilde B(F)\coloneqq \frac{1}{n}\int_{X^n}\Pi(x) d(dF)^{\otimes n}(x), \qquad \forall \ F\in\mathcal F.\]
Recall that for \(F\in\mathcal F\) and \(x\in X\),
\[\tag{\(C2\)}\label{eq:a2} A_F(x)\coloneqq \int_{X^{n-1}} \Pi\left(x,x_2,\ldots,x_n\right) d(dF)^{\otimes(n-1)}\left(x_2,\ldots,x_n\right).\]
We now compute the directional derivative of \(\widetilde B\):
\begin{lemma}\label{lem:dir-deriv-Btilde}
For any \(F,G\in\mathcal F\), \(\delta \widetilde B(F;G-F)=\int_X A_F(x) d(G-F)(x)\).\end{lemma}

\begin{proof}
Fix \(F,G\in\mathcal F\) and let \(\nu\coloneqq d(G-F)\), a finite signed measure on \(X\). For \(\varepsilon\in\left[0,1\right]\), write \(F_\varepsilon\coloneqq (1-\varepsilon)F+\varepsilon G\), so \(dF_\varepsilon=dF+\varepsilon \nu\).

By multilinearity of the product measure,
\[(dF_\varepsilon)^{\otimes n}=(dF+\varepsilon \nu)^{\otimes n} = (dF)^{\otimes n} + \varepsilon\sum_{i=1}^n (dF)^{\otimes(i-1)}\otimes \nu \otimes (dF)^{\otimes(n-i)} + O(\varepsilon^2),\]
where \(O(\varepsilon^2)\) denotes a signed measure whose total variation is \(O(\varepsilon^2)\) as \(\varepsilon\downarrow 0\). As \(\Pi\) is bounded, dividing by \(\varepsilon\) and letting \(\varepsilon\downarrow 0\) yields
\[\delta \widetilde B(F;G-F) = \frac{1}{n}\sum_{i=1}^n \int_{X^n}\Pi(x) d\left((dF)^{\otimes(i-1)}\otimes \nu \otimes (dF)^{\otimes(n-i)}\right)(x).\]
By symmetry of \(\Pi\), each term in the sum is equal, so the \(1/n\) cancels the sum, so
\[\delta \widetilde B(F;G-F) = \int_X \left( \int_{X^{n-1}} \Pi\left(x,x_2,\ldots,x_n\right) d(dF)^{\otimes(n-1)}\left(x_2,\ldots,x_n\right) \right) \nu(dx) = \int_X A_F(x) d(G-F)(x),\]
as claimed.
\end{proof}

\subsection{Proof of \texorpdfstring{\Cref{lem:game-kkt}}{Lemma~\ref{lem:game-kkt}}}
\begin{proof}[Proof of \Cref{lem:game-kkt}]
Suppose first that \(F\) is a symmetric equilibrium. Fix opponents at \(F\). By the definition of \(a_F\), for any deviation \(H\in\mathcal F\),
\[u_1\left(H,F,\ldots,F\right)=\int_X a_F(x) dH(x)-C(H).\]
As \(F\) is a best response to \((F,\ldots,F)\), it maximizes \(J_{a_F}(H)\coloneqq \int_X a_F(x) dH(x)-C(H)\) over \(H\in\mathcal F\). \Cref{lem:kkt_simplex}, therefore, yields \(\lambda_g\in\R\) such that, writing \(h(x)\coloneqq a_F(x)-c_F(x)\),
\[h(x)\le \lambda_g\ \ \forall x\in X, \qquad \text{and} \qquad \int_X \left(\lambda_g-h(x)\right) dF(x)=0.\]
It remains to show equality on \(\supp(dF)\). Fix \(x_0\in\supp(dF)\). If \(dF(\{x_0\})>0\), equality holds by \Cref{lem:compl_support}\ref{it:b41}. If \(dF(\{x_0\})=0\), then \Cref{lem:aF-cont} implies \(a_F(\cdot)\) is continuous at \(x_0\), and \(c_F(\cdot)\) is continuous by \Cref{ass:cost}\ref{ass2}, so \(h(\cdot)\) is continuous at \(x_0\). Since \(x_0\in\supp(dF)\), \Cref{lem:compl_support}\ref{it:b42} yields \(h(x_0)=\lambda_g\).

Conversely, suppose \(F\) satisfies the stated conditions. Since \(\Phi(F,x)=\lambda_g\) for every \(x\in\supp(dF)\) and \(F\) assigns probability one to its support, \(\int_X\Phi(F,x)dF(x)=\lambda_g\). Moreover, \(\Phi(F,x)\le\lambda_g\) for every \(x\in X\), so \(\int_X\Phi(F,x)dH(x)\le\lambda_g\) for every \(H\in\mathcal F\), since each \(H\) is a probability distribution. Convexity of \(C\) implies \(C(H)-C(F)\ge\int_Xc_F(x)d(H-F)(x)\), so
\[u_1(H,F,\ldots,F)-u_1(F,\ldots,F) \le\int_X\Phi(F,x)d(H-F)(x)\le0.\]
Thus, \(F\) is a best response to \((F,\ldots,F)\), so \((F,\ldots,F)\) is a Nash equilibrium.
\end{proof}

\subsection{Proof of \texorpdfstring{\Cref{lem:planner-kkt}}{Lemma~\ref{lem:planner-kkt}}}
\begin{proof}[Proof of \Cref{lem:planner-kkt}]
By \Cref{lem:dir-deriv-Btilde}, the directional derivative of the planner prize functional
\[\widetilde B(F)\coloneqq \frac{1}{n}\int_{X^n}\Pi(x) d(dF)^{\otimes n}(x)\]
satisfies
\[\delta \widetilde B(G;F-G)=\int_X A_G(x) d(F-G)(x), \qquad \forall \ F\in\mathcal F.\]
Since \(G\) maximizes \(v(F)=\widetilde B(F)-C(F)\), our variational inequality is
\[\int_X \left(A_G(x)-c_G(x)\right) d(F-G)(x)\le 0, \qquad \forall \ F\in\mathcal F.\]

Moreover, this variational inequality implies that \(G\) maximizes the auxiliary objective \(\int_X A_G(t) dF(t)-C(F)\) over \(\mathcal F\).\footnote{Indeed, rewriting this as \(J_{A_G}\), for any \(F\in\mathcal F\), simply use the convexity of \(C\) to get \(J_{A_G}(F)-J_{A_G}(G) \le \int_X \left(A_G(x)-c_G(x)\right)d(F-G)(x) \le 0\), i.e., \(G\) maximizes \(J_{A_G}\).}

Applying \Cref{lem:kkt_simplex} to the bounded Borel function \(A_G(x)\) yields \(\lambda_p\) with \(A_G-c_G\le \lambda_p\) and complementarity \(\int_X \left(\lambda_p-(A_G-c_G)\right)dG=0\). Since \(\Pi\) is continuous on the compact \(X^n\), \(A_G(x)\) is continuous on \(X\) by dominated convergence. Thus, \(A_G(\cdot)-c_G(\cdot)\) is continuous, and \Cref{lem:compl_support}\ref{it:b42} implies equality on \(\supp(dG)\).

Conversely, suppose \(v\) is concave and \(G\) satisfies the stated conditions. Since \(A_G(x)-c_G(x)\le\lambda_p\) for every \(x\in X\), with equality for every \(x\in\supp(dG)\), concavity of \(v\) yields
\[v(F)-v(G)\le\delta v(G;F-G) =\int_X\left(A_G(x)-c_G(x)\right)d(F-G)(x)\le0\]
for every \(F\in\mathcal F\). Thus, \(G\) maximizes \(v\). Strict concavity of \(v\) implies that its maximizer is unique.
\end{proof}

\section{Rank-Order Contests Proofs}\label[appendix]{ap:roproofs}
Here we prove \Cref{lem:roc} and \Cref{thm:icx_rank_order}.
\begin{proof}[Proof of \Cref{lem:roc}]
We prove the result for \(\mathbf{v}\), as \(\mathbf{w}\)'s result follows analogously.

As \(v_1>v_2\), escaping a tie for the top rank yields a discrete gain on any tie-for-top event. Thus, \Cref{ass:strict-tie} holds, so any symmetric equilibrium is atomless on \(\left(0,1\right)\). Also, \Cref{lem:no-atom-1} implies \(dF\left( \left\{1\right\} \right)=0\).

We now show \(dF\left( \left\{0\right\} \right)=0\). Suppose, toward a contradiction, that \(dF\left( \left\{0\right\} \right)=\alpha>0\). Consider a deviation that removes an \(\varepsilon\)-slice of this atom and redistributes it uniformly on \(\left(0,\delta\right]\), with \(\varepsilon \in \left(0,\alpha\right)\) and \(\delta>0\) small. On the event that every opponent draws \(0\), which has probability \(\alpha^{n-1}>0\), the deviator moves from the tie outcome at \(0\) to being the unique winner. Under uniform tie-breaking, the deviator's expected prize at a full tie is \(\frac{1}{n}\sum_{k=1}^n v_k=\frac{1}{n}\), whereas as unique winner it is \(v_1\). Since \(v_n=0\) and \(\sum_k v_k=1\), we have \(v_1>\frac{1}{n}\), so the expected prize gain is at least \(\varepsilon \alpha^{n-1}\left( v_1-\frac{1}{n} \right)>0\).

Of course, the cost changes as well, so we appeal to the G\^ateaux differentiability and continuity of \(c_F\left( x \right)\) in \(x\): the first-order cost increase per unit shifted mass equals the difference between the average of \(c_{F}\left( x \right)\) on \(\left(0,\delta\right]\) and \(c_{F}\left( 0 \right)\), which can be made arbitrarily small by choosing \(\delta\) small. Hence, for \(\delta\) sufficiently small and then \(\varepsilon\) small, the deviation is profitable, contradicting equilibrium. We conclude that \(dF\left( \left\{0\right\} \right)=0\).

Furthermore, we know from \Cref{lem:game-kkt} that there exists \(\lambda_F\) such that \(a_F\left(x\right)-c_F\left(x\right)\le\lambda_F\) for all \(x\in\left[0,1\right]\), with equality for all \(x\in\operatorname{supp}\left(dF\right)\). Our next step is to argue \(0 \in \operatorname{supp}\left( dF \right)\). Let \(\ubar{x} \coloneqq \inf \operatorname{supp}\left( dF \right)\). Suppose for the sake of contradiction that \(\ubar{x}>0\). Since \(\operatorname{supp}\left(dF\right)\) is closed, \(\ubar{x}\in\operatorname{supp}\left(dF\right)\). No opponent draws below \(\ubar{x}\), and by \Cref{lem:no-interior-atoms,lem:no-atom-1}, \(dF\left(\left\{\ubar{x}\right\}\right)=0\). Consequently, a deviator choosing \(x=\ubar{x}\) is last almost surely and receives \(v_n=0\). A deviator choosing \(x=0\) is also last almost surely and receives \(v_n=0\). Thus, \(a_F\left(\ubar{x}\right)=a_F\left(0\right)=0\).

Since \(\ubar{x}\in\operatorname{supp}\left(dF\right)\), equality in \Cref{lem:game-kkt} produces
\[\lambda_F = a_F\left(\ubar{x}\right)-c_F\left(\ubar{x}\right) = -\Upsilon'\left(z_F\right)\gamma\left(\ubar{x}\right) < 0,\]
where the strict inequality uses \Cref{ass:contest}: \(\gamma\left(\ubar{x}\right)>0\) and \(\Upsilon'\left(z_F\right)\ge \Upsilon'\left(0\right)>0\). But the global inequality in \Cref{lem:game-kkt}, evaluated at \(0\), yields \(0 = a_F\left(0\right)-c_F\left(0\right) \le \lambda_F\), a contradiction. Therefore, \(\ubar{x}=0\), and \(0 \in \operatorname{supp}\left( dF \right)\).

Since \(dF\left( \left\{0\right\} \right)=0\), there is no tie at \(0\) with positive probability, so a deviator choosing \(x=0\) is last almost surely and receives \(v_n=0\), whence \(a_{F}\left( 0 \right)=0\). As \(0 \in \operatorname{supp}\left( dF \right)\), equality in \Cref{lem:game-kkt} produces \(\lambda_F = a_{F}\left( 0 \right) - c_{F}\left( 0 \right) = 0\).

Now take \(x \in \operatorname{supp}\left( dF \right)\). Since \(dF\left( \left\{x\right\} \right)=0\), ties at \(x\) occur with probability \(0\). Consequently, the probability that a player with realized performance \(x\) is ranked \(k\)-th equals \(\binom{n-1}{k-1} F\left( x \right)^{n-k}\left( 1-F\left( x \right) \right)^{k-1}\), and so \(a_{F}\left( x \right) = \Psi\left( F\left( x \right); \mathbf{v}\right)\). Plugging into the first-order conditions and using \(\lambda_F=0\) produces, for \(x \in \operatorname{supp}\left( dF \right)\),
\[\Psi\left( F\left( x \right); \mathbf{v}\right) = \Upsilon'\left(z_F\right)\gamma\left( x \right).\]

Let \(X\sim F\) and define \(U\coloneqq F\left(X\right)\). Because \(F\) is continuous, \(U\sim \mathrm{Unif}\left(\left[0,1\right]\right)\). \(X\in\operatorname{supp}\left(dF\right)\) almost surely implies \(\Upsilon'\left(z_F\right)\gamma\left(X\right)=\Psi\left(U;\mathbf{v}\right)\) almost surely. Taking expectations, we have
\[\Upsilon'\left(z_F\right)z_F = \int_0^1\Psi\left(q;\mathbf{v}\right)dq.\]
Moreover,
\[\int_0^1\binom{n-1}{k-1}q^{n-k}\left(1-q\right)^{k-1}dq = \frac{1}{n}\]
for every \(k\), so \(\int_0^1\Psi\left(q;\mathbf{v}\right)dq = \frac{1}{n}\). Therefore, \(z_F\Upsilon'\left(z_F\right)=\frac{1}{n}\). The solution to this equation, \(z_n^\ast\), is well-defined; and we define \(\Upsilon'\left(z_n^\ast\right)\eqqcolon \alpha_n^\ast\). Consequently, \(\gamma\left(X\right) = \frac{\Psi\left(U;\mathbf{v}\right)}{\alpha_n^\ast}\) almost surely.

Finally, we prove uniqueness. Let \(\tilde{F}\) be any symmetric equilibrium, and let \(\tilde{X}\sim \tilde{F}\). Applying the preceding argument to \(\tilde{F}\), with \(\tilde{U}\coloneqq\tilde{F}\left(\tilde{X}\right)\sim\mathrm{Unif}\left(\left[0,1\right]\right)\), yields
\[\gamma\left(\tilde{X}\right) = \frac{\Psi\left(\tilde{U};\mathbf{v}\right)}{\alpha_n^\ast}\]
almost surely. Hence, \(\gamma\left(X\right)\) and \(\gamma\left(\tilde{X}\right)\) have the same distribution. Since \(\gamma\) is strictly increasing, this implies \(X\) and \(\tilde{X}\) have the same distribution, so \(\tilde{F}=F\).
\end{proof}

Now we prove the theorem.
\begin{proof}[Proof of \Cref{thm:icx_rank_order}]
Let \(X \sim F\) and define \(U \coloneqq F\left( X \right)\). Since \(F\) is continuous, \(U \sim \mathrm{Unif}\left( \left[0,1\right] \right)\). By the proof of \Cref{lem:roc},
\[\gamma\left( X \right) = g_v\left(U\right) \quad\text{a.s., where}\quad g_v\left(q\right)\coloneqq\frac{\Psi\left(q;\mathbf{v}\right)}{\alpha_n^\ast}.\]
Do analogously for \(\mathbf{w}\) so that
\[\gamma\left(Y\right)=g_w\left(U\right) \quad\text{a.s., where}\quad g_w\left(q\right)\coloneqq\frac{\Psi\left(q;\mathbf{w}\right)}{\alpha_n^\ast}.\]
Since \(v_1>v_2\) and \(w_1>w_2\), both \(g_v\) and \(g_w\) are strictly increasing, and so they are the quantile functions of \(Z\coloneqq\gamma\left(X\right)\) and \(T\coloneqq\gamma\left(Y\right)\), respectively.

Define
\[p_k\left( q \right) \coloneqq \binom{n-1}{k-1} q^{n-k}\left( 1-q \right)^{k-1}.\]
Note that \(\Psi\left( q;\mathbf{v} \right) = \sum_{k=1}^n v_k p_k\left( q \right)\) for all \(q \in \left[0,1\right]\). Hence,
\[g_w\left( q \right) - g_v\left( q \right) = \frac{\Psi\left( q;\mathbf{w} \right) - \Psi\left( q;\mathbf{v} \right)}{\alpha_n^\ast} = \frac{1}{\alpha_n^\ast}\sum_{k=1}^n \left( w_k - v_k \right)p_k\left( q \right).\]
Because \(\mathbf{w}\) majorizes \(\mathbf{v}\), \(\mathbf{w}\) can be obtained from \(\mathbf{v}\) via a finite sequence of reverse Pigou-Dalton transfers. Since the map \(\mathbf{v} \mapsto \Psi\left( q;\mathbf{v} \right)\) is linear for each \(q\), and the map \(h \mapsto \int_0^p h\left( q \right) dq\) is linear for each \(p\), it suffices to prove the claim for a single reverse Pigou-Dalton transfer and then sum the resulting integrated inequalities across the finite sequence. So fix \(i<j\) and \(\delta>0\) and suppose \(\mathbf{w}\) is obtained from \(\mathbf{v}\) by moving \(\delta\) from rank \(j\) to rank \(i\):
\[w_i = v_i + \delta,\quad w_j = v_j - \delta,\quad\text{and}\quad w_k=v_k\ \text{for all \(k\neq i,j\).}\]
Then,
\[\Psi\left( q;\mathbf{w} \right) - \Psi\left( q;\mathbf{v} \right) = \delta\left( p_i\left( q \right) - p_j\left( q \right) \right),\]
so
\[g_w\left( q \right) - g_v\left( q \right) = \frac{\delta}{\alpha_n^\ast}\left( p_i\left( q \right) - p_j\left( q \right) \right).\]

For \(q \in \left(0,1\right)\),
\[\frac{p_i\left( q \right)}{p_j\left( q \right)} = \frac{\binom{n-1}{i-1}}{\binom{n-1}{j-1}} \left( \frac{q}{1-q} \right)^{j-i},\]
which is strictly increasing in \(q\). Moreover, this ratio converges to \(0\) as \(q \downarrow 0\) and to \(+\infty\) as \(q \uparrow 1\). Therefore, \(p_i\left( q \right)-p_j\left( q \right)\) crosses \(0\) exactly once, from negative to positive, and so does \(g_w\left( q \right)-g_v\left( q \right)\).

Next, compute the mean:
\[\int_0^1 p_k\left( q \right) dq = \binom{n-1}{k-1}\int_0^1 q^{n-k}\left( 1-q \right)^{k-1}dq = \binom{n-1}{k-1}\mathrm{B}\left( n-k+1,k \right) = \frac{1}{n}.\]
Consequently,
\[\int_0^1 \Psi\left( q;\mathbf{v} \right) dq = \sum_{k=1}^n v_k \int_0^1 p_k\left( q \right) dq = \frac{1}{n}\sum_{k=1}^n v_k = \frac{1}{n}.\]
Therefore,
\[\mathbb{E}\left[ Z \right] = \int_0^1 g_v\left( q \right) dq = \frac{1}{\alpha_n^\ast n} = \int_0^1 g_w\left( q \right) dq = \mathbb{E}\left[ T \right].\]
So \(Z\) and \(T\) have equal means.

Let \(H\left( p \right) \coloneqq \int_0^p \left( g_w\left( q \right)-g_v\left( q \right) \right) dq\). Since \(g_w-g_v\) crosses \(0\) exactly once from negative to positive and \(\int_0^1 \left( g_w-g_v \right) dq=0\), it follows that \(H\left( p \right) \le 0\) for all \(p \in \left[0,1\right]\), with equality at \(p=1\). By the integrated-quantile characterization of the convex order, this implies \(T \ge_{\mathrm{cx}} Z\).

Finally, because \(\gamma\) is strictly increasing and concave, \(\gamma^{-1}\) exists and is strictly increasing and convex. Let \(u\) be any increasing convex function. Then \(u \circ \gamma^{-1}\) is convex, so \(T \ge_{\mathrm{cx}} Z\) implies
\[\mathbb{E}\left[ u\left( Y \right) \right] = \mathbb{E}\left[ u\left( \gamma^{-1}\left( T \right) \right) \right] \geq \mathbb{E}\left[ u\left( \gamma^{-1}\left( Z \right) \right) \right] = \mathbb{E}\left[ u\left( X \right) \right].\]
Therefore, \(Y \ge_{\mathrm{icx}} X\).
\end{proof}

\section{R\&D Proofs}\label[appendix]{rdproofs}

\subsection{Auxiliary Lemmas}

In the terminology of the R\&D section, \eqref{eq:a1} and \eqref{eq:a2} are
\[\widetilde{B}\left(F\right)\coloneqq \frac{1}{n}\int_{\left[0,\infty\right]^n} V\left(\min_{1\le j\le n} t_j\right)  d\left(dF\right)^{\otimes n}\left(t\right), \qquad \forall \ F\in\mathcal F,\]
and
\[A_F\left(t\right)\coloneqq \int_{\left[0,\infty\right]^{n-1}} V\left(\min\left\{t,t_2,\dots,t_n\right\}\right)  d\left(dF\right)^{\otimes\left(n-1\right)}\left(t_2,\dots,t_n\right).\]
Equivalently, if \(T_2,\dots,T_n\sim dF\) are i.i.d. and \(M_F\coloneqq \min\left\{T_2,\dots,T_n\right\}\), then \(A_F\left(t\right)=\mathbb{E}\left[V\left(\min\left\{t,M_F\right\}\right)\right]\).

Likewise from \Cref{lem:dir-deriv-Btilde}, we record the directional derivative of \(\widetilde B\): for any \(F,G\in\mathcal F\),
\[\delta \widetilde{B}\left(F;G-F\right)=\int_{\left[0,\infty\right]} A_F\left(t\right)  d\left(G-F\right)\left(t\right).\]
For \(t\in\left[0,\infty\right)\) define
\[\widetilde{V}_F\left(t\right)\coloneqq \mathbb{E}\left[V\left(M_F\right)\mathbf{1}\left\{M_F>t\right\}\right],\]
with the convention \(\widetilde{V}_F\left(\infty\right)\coloneqq 0\).

This allows us to rewrite \(A_F\):
\begin{lemma}\label{lem:B7}
For all \(t\in\left[0,\infty\right]\), \(A_F\left(t\right)=A_F\left(\infty\right)+V\left(t\right)\left(1-F\left(t\right)\right)^{n-1}-\widetilde{V}_F\left(t\right)\).
\end{lemma}

\begin{proof}[Proof of \Cref{lem:B7}]
Directly, for all \(t < \infty\),
\[A_F\left(t\right)=\mathbb{E}\left[V\left(\min\left\{t,M_F\right\}\right)\right] =V\left(t\right)\mathbb{P}\left(M_F>t\right)+\mathbb{E}\left[V\left(M_F\right)\mathbf{1}\left\{M_F\le t\right\}\right].\]
Moreover, \(A_F\left(\infty\right)=\mathbb{E}\left[V\left(\min\left\{\infty,M_F\right\}\right)\right]=\mathbb{E}\left[V\left(M_F\right)\right]\). Therefore,
\[A_F\left(t\right)=A_F\left(\infty\right)+V\left(t\right)\mathbb{P}\left(M_F > t\right)-\mathbb{E}\left[V\left(M_F\right)\mathbf{1}\left\{M_F > t\right\}\right] =A_F\left(\infty\right)+V\left(t\right)\mathbb{P}\left(M_F > t\right)-\widetilde{V}_F \left(t\right).\]
Since \(T_2,\dots,T_n\) are i.i.d. with cdf \(F\), \(\mathbb{P}\left(M_F > t\right)=\left(1-F\left(t\right)\right)^{n-1}\). The case \(t=\infty\) holds by the definition \(\widetilde{V}_F\left(\infty\right)=0\).
\end{proof}

Next, we rewrite the planner's reduced-form marginal benefit as a tail integral:
\begin{lemma}\label{lem:rd_planner_tail}
For every \(F\in\mathcal F\) and every \(t<\infty\), \(V\left(t\right)\left(1-F\left(t\right)\right)^{n-1}-\widetilde{V}_F\left(t\right) = \int_t^\infty rV\left(s\right)\left(1-F\left(s\right)\right)^{n-1}ds\).
\end{lemma}

\begin{proof}[Proof of \Cref{lem:rd_planner_tail}]
By definition,
\[V\left(t\right)\left(1-F\left(t\right)\right)^{n-1}-\widetilde{V}_F\left(t\right) = \mathbb{E}\left[\left(V\left(t\right)-V\left(M_F\right)\right)\mathbf{1}\left\{M_F>t\right\}\right].\]
On \(\left\{M_F>t\right\}\), \(V\left(t\right)-V\left(M_F\right)=\int_t^{M_F}rV\left(s\right)ds\). Tonelli's theorem allows us to exchange expectation and integration, so that
\[\mathbb{E}\left[\mathbf{1}\left\{M_F>t\right\}\int_t^{M_F}rV\left(s\right)ds\right] = \int_t^\infty rV\left(s\right)\mathbb{P}\left(M_F>s\right)ds = \int_t^\infty rV\left(s\right)\left(1-F\left(s\right)\right)^{n-1}ds.\qedhere\]
\end{proof}

\subsection{Proof of \texorpdfstring{\Cref{cor:rd_overinvestment}}{Theorem~\ref{cor:rd_overinvestment}}}

\begin{proof}[Proof of \Cref{cor:rd_overinvestment}]
Define (understood pointwise)
\[D \coloneqq \max\left\{G-F,0\right\}, \qquad F^+\coloneqq \max\left\{F, G\right\}, \qquad\text{and}\qquad G^-\coloneqq \min\left\{F,G\right\};\]
so that \(F^+ - F = D\) and \(G^- - G = - D\).

Since \(F\) is a symmetric equilibrium, it maximizes \(\int_{\left[0,\infty\right]}B_F\left(t\right)dH\left(t\right)-C\left(H\right)\) over \(H\in\mathcal F\). Accordingly, a deviation to \(F^+\) yields
\[\int_{\left[0,\infty\right]}\left(B_F\left(t\right)-c_F\left(t\right)\right)d\left(F^+-F\right)\left(t\right)\le0.\]
Because \(B_F\left(t\right)=V\left(t\right)\left(1-F\left(t\right)\right)^{n-1}\) is decreasing and \(B_F\left(\infty\right)=0\), integration by parts yields
\[\int_{\left[0,\infty\right]}B_F\left(t\right)d\left(F^+-F\right)\left(t\right) = \int_0^\infty D\left(t\right)\left[-dB_F\left(t\right)\right].\]
Moreover,
\[-dB_F\left(t\right)\ge rV\left(t\right)\left(1-F\left(t\right)\right)^{n-1}dt.\]
Recalling the normalized marginal cost delivered by \Cref{ass:rd_monotone_marginal_cost} (\(c_H\left(\tau\right)\coloneqq \int_{\left[\tau,\infty\right]}m_H\left(t\right)d\eta\left(t\right)\) for all \(\tau < \infty\)),\[\int_{\left[0,\infty\right]}c_F\left(t\right)d\left(F^+-F\right)\left(t\right) = \int_{\left[0,\infty\right]}D\left(t\right)m_F\left(t\right)d\eta\left(t\right).\] Therefore,
\[\int_0^\infty D\left(t\right)rV\left(t\right)\left(1-F\left(t\right)\right)^{n-1}dt \le \int_{\left[0,\infty\right]}D\left(t\right)m_F\left(t\right)d\eta\left(t\right). \tag{\(E1\)} \label{eq:d1}\]

Now consider the planner. Since \(G\) maximizes \(v\), we must have
\[\int_{\left[0,\infty\right]}\left(A_G\left(t\right)-c_G\left(t\right)\right)d\left(G^- -G\right)\left(t\right)\le0.\]
Subtracting \(A_G\left(\infty\right)\) from the integrand does not change this inequality, since \(G^- -G\) has total mass zero. \Cref{lem:B7} reveals the planner's marginal benefit:
\[A_G\left(t\right)-A_G\left(\infty\right) = V\left(t\right)\left(1-G\left(t\right)\right)^{n-1}-\widetilde{V}_G\left(t\right).\]
By \Cref{lem:rd_planner_tail},
\[A_G\left(t\right)-A_G\left(\infty\right) = \int_t^\infty rV\left(s\right)\left(1-G\left(s\right)\right)^{n-1}ds \quad \forall t < \infty,\]
and \(0\) at \(t=\infty\). Since \(G^- -G=-D\), we use Fubini's theorem to get
\[\int_{\left[0,\infty\right]}\left(A_G\left(t\right)-A_G\left(\infty\right)\right)d\left(G^- -G\right)\left(t\right) = -\int_0^\infty D\left(t\right)rV\left(t\right)\left(1-G\left(t\right)\right)^{n-1}dt.\]
Using the marginal-cost form provided by \Cref{ass:rd_monotone_marginal_cost}, as above,
\[\int_{\left[0,\infty\right]}c_G\left(t\right)d\left(G^- -G\right)\left(t\right) = -\int_{\left[0,\infty\right]}D\left(t\right)m_G\left(t\right)d\eta\left(t\right).\]
Thus,
\[\int_{\left[0,\infty\right]}D\left(t\right)m_G\left(t\right)d\eta\left(t\right) \le \int_0^\infty D\left(t\right)rV\left(t\right)\left(1-G\left(t\right)\right)^{n-1}dt. \tag{\(E2\)}\label{eq:d2}\]

Combining \eqref{eq:d1} and \eqref{eq:d2}, then appealing to \Cref{ass:rd_monotone_marginal_cost}, yields
\[\int_0^\infty D\left(t\right)rV\left(t\right)\left(\left(1-F\left(t\right)\right)^{n-1}-\left(1-G\left(t\right)\right)^{n-1}\right)dt \le \int_{\left[0,\infty\right]}D\left(t\right)\left(m_F\left(t\right)-m_G\left(t\right)\right)d\eta\left(t\right) \leq 0. \tag{\(E.3\)}\label{eq:d3}\]

Suppose for the sake of contradiction that \(D\) is not identically zero, i.e., that there exists \(t<\infty\) such that \(G\left(t\right)>F\left(t\right)\). By right-continuity, this holds on a set of positive Lebesgue measure. On this set,
\[rV\left(t\right)\left(\left(1-F\left(t\right)\right)^{n-1}-\left(1-G\left(t\right)\right)^{n-1}\right)>0.\]
Thus, the left-hand side of \eqref{eq:d3} is strictly positive, a contradiction. Therefore, \(D\left(t\right)=0\) for every finite \(t\), so \(G\left(t\right)\le F\left(t\right)\) for every \(t<\infty\). Since \(F\left(\infty\right)=G\left(\infty\right)=1\), the inequality also holds at \(t=\infty\).
\end{proof}

\section{Online Appendix}\label[appendix]{app:supplement}

\subsection{Alternative Feasible Sets}
\label{subsec:alt_feasible_sets}

In some applications, players are not free to choose any \(F\in\mathcal{F}\), but are instead restricted to a subset \(\mathcal{F}_0\subseteq\mathcal{F}\). This can capture, for instance, moment constraints (e.g., a fixed mean) or feasibility constraints defined by mean-preserving contractions relative to a baseline distribution. In this subsection, I record a sufficient condition under which the equilibrium existence result of \Cref{prop:symNE} extends verbatim.

Fix a nonempty set \(\mathcal{F}_0\subseteq\mathcal{F}\). Consider the same \(n\)-player game as in \S2 of the main text, except that each player \(i\) must choose \(F_i\in\mathcal{F}_0\).

\begin{assumption}\label{ass:restricted_feasible}
Fix \(\mathcal{F}_0\subseteq\mathcal{F}\).
\begin{enumerate}[nosep,noitemsep]
\item\label{it:r1} \(\mathcal{F}_0\) is compact and convex with respect to weak convergence of the induced measures on \(X\).
\item\label{it:r2} For every \(F\in\mathcal{F}_0\), every \(H_0\in\mathcal{F}_0\), and every \(\varepsilon>0\), there exists \(\tilde H\in\mathcal{F}_0\) such that \(d\tilde H(\{x\})=0\) for every \(x\in X\) and \(u_1\left(\tilde H,F,\dots,F\right)\ge u_1\left(H_0,F,\dots,F\right)-\varepsilon\).
\end{enumerate}
\end{assumption}

\begin{proposition}\label{prop:restricted_existence}
Posit \Cref{ass:prize,ass:cost,ass:restricted_feasible}. Then the restricted game with strategy space \(\mathcal{F}_0\) admits a symmetric pure-strategy Nash equilibrium \(F^\ast\in\mathcal{F}_0\).
\end{proposition}
The proof of \Cref{prop:restricted_existence} is analogous to that of \Cref{prop:symNE}, with \(\mathcal{F}\) replaced by \(\mathcal{F}_0\), so we merely sketch it.
\begin{proof}[Proof of \Cref{prop:restricted_existence}]
By \Cref{ass:restricted_feasible}(\ref{it:r1}), the symmetric strategy space \(\mathcal{F}_0\) is a compact metric space. Diagonal quasiconcavity holds on \(\mathcal{F}_0\) because for fixed opponents \(F\), the map \(H\mapsto u_1\left(H,F,\dots,F\right)\) is concave on \(\mathcal{F}\) (the expected-prize term is affine in \(dH\) and \(-C\) is concave by convexity of \(C\)), and \(\mathcal{F}_0\) is convex.

For diagonal better-reply security, fix \(F\in\mathcal{F}_0\) such that \(\left(F,\dots,F\right)\) is not a symmetric equilibrium. Then there exist \(H_0\in\mathcal{F}_0\) and \(\eta>0\) with \(u_1\left(H_0,F,\dots,F\right)\ge u_1\left(F,\dots,F\right)+\eta\). By \Cref{ass:restricted_feasible}(\ref{it:r2}), choose an atomless \(\tilde H\in\mathcal{F}_0\) such that
\[u_1\left(\tilde H,F,\dots,F\right)\ge u_1\left(H_0,F,\dots,F\right)-\eta/2,\]
and so \(u_1\left(\tilde H,F,\dots,F\right)\ge u_1\left(F,\dots,F\right)+\eta/2\). \Cref{lem:atomless_cont} implies the map \(G\mapsto u_1\left(\tilde H,G,\dots,G\right)\) is continuous at \(F\), so \(\tilde H\) secures strictly above \(u_1\left(F,\dots,F\right)\) in a neighborhood along the diagonal. \citet[Theorem 4.1]{reny1999existence}, therefore, implies the existence of a symmetric equilibrium in \(\mathcal{F}_0\).
\end{proof}

\subsubsection{Examples.}
\begin{enumerate}[noitemsep]
\item[(i)] \emph{Fixed mean.} Fix \(m\in(0,1)\) and define \(\mathcal{F}(m)\coloneqq\left\{F\in\mathcal{F}:\int_X x dF(x)=m\right\}\). It is a standard fact that \(\mathcal{F}(m)\) is a compact convex subset of \(\mathcal{F}\). Moreover, \(\mathcal{F}(m)\) satisfies \Cref{ass:restricted_feasible}(\ref{it:r2}): one can split atoms into small absolutely continuous components and offset any induced mean change by transferring an arbitrarily small amount of mass between two interior intervals. Boundedness of \(\pi_1\) and continuity of \(C\) then ensure the atomless perturbation can be chosen to change \(u_1\left(\cdot,F,\dots,F\right)\) by at most \(\varepsilon\).

\item[(ii)] \emph{Mean-preserving contractions of a baseline distribution.} Fix \(G\in\mathcal{F}\). For \(F\in\mathcal{F}\), write \(F\preceq_{cx}G\) if \(G\) dominates \(F\) in the convex order, and define \(\mathcal{F}_{cx}(G)\coloneqq\left\{F\in\mathcal{F}\colon F\preceq_{cx}G\right\}\). It is a standard fact that \(\mathcal{F}_{cx}(G)\) is compact and convex. Finally, if \(\tilde H\preceq_{cx}H_0\) and \(H_0\preceq_{cx}G\), then \(\tilde H\preceq_{cx}G\). Thus, we can choose the smoothing in \Cref{ass:restricted_feasible}(\ref{it:r2}) to be a sequence of atomless mean-preserving contractions of \(H_0\), which remain feasible.
\end{enumerate}

\Cref{ass:restricted_feasible}(\ref{it:r2}) rules out restrictions that force atoms (e.g., requiring support on a finite grid), in which case tie discontinuities cannot generally be smoothed away.

\subsection{Price \& Quality Competition}

This appendix extends \Cref{cor:price_equals_mc_gap,cor:necessary_mc_limit} to allow i. mixed strategies over prices, and ii. joint randomization over \((F,p)\). Fix \(n\ge2\) and maintain \Cref{as:price_smoothing}, repeated here for clarity:

\paragraph{Assumption 6.1} Take some \(\sigma\in(0,1)\). After a firm draws \(Q_i\sim dF_i\), the representative consumer draws an independent
\(\varepsilon_i\) with cdf \(G\) on \([0,1]\) that admits a continuous density \(g\). The consumer's value for product \(i\) is \(\hat{Q}_i \coloneqq (1-\sigma)Q_i+\sigma \varepsilon_i\in[0,1]\).

\bigskip

Let \(X\coloneqq[0,1]\) and \(P\coloneqq[0,\bar p]\). For any \(F\in\mathcal{F}\), let \(\hat{F}\) denote the induced cdf of \(\hat Q\coloneqq(1-\sigma)Q+\sigma\varepsilon\). As in \secref{sec:unit_demand_prices}, we extend \(\hat{F}(x)\coloneqq0\) for \(x<0\) and \(\hat{F}(x)\coloneqq1\) for \(x>1\).

Let \(M\) be a Borel probability measure on \(\mathcal{F}\times P\). Under symmetric play, each opponent draws \((F_j,p_j)\sim M\) i.i.d., then draws \(Q_j\sim dF_j\) and \(\varepsilon_j\sim G\) i.i.d., and sets \(\hat Q_j\coloneqq(1-\sigma)Q_j+\sigma\varepsilon_j\). Define opponent net utility \(U_j\coloneqq\hat Q_j-p_j\). Let \(\tilde F_M\colon\mathbb{R}\to[0,1]\) denote the cdf of \(U_j\):
\[\tilde F_M(z)\coloneqq\int_{\mathcal{F}\times P}\hat{F}\left(z+p\right) dM(F,p).\]
Equivalently, \(\tilde F_M(z)=0\) for \(z<-\bar p\) and \(\tilde F_M(z)=1\) for \(z>1\).

Fix a deviator price \(r\in P\). Conditional on realizing quality \(q\in X\) and taste shock \(e\in[0,1]\), the deviator wins if and only if \((1-\sigma)q+\sigma e-r>\max_{2\le j\le n}U_j \). Accordingly, its interim expected profit is
\[a_M(q;r)\coloneqq r\cdot\int_0^1\left(\tilde F_M\left((1-\sigma)q+\sigma e-r\right)\right)^{n-1}dG(e).\]
For each \(r\in P\), denote
\[W_M(r)\coloneqq\max_{H\in\mathcal{F}}\left(\int_X a_M(q;r) dH(q)-C(H)\right).\]

\begin{proposition}\label{prop:mixed_s7:FOC}
Posit \Cref{ass:cost,as:price_smoothing}. Fix a Borel probability measure \(M\) on \(\mathcal{F}\times P\) and a price \(r\in P\). Let \(F\in\mathcal{F}\) attain the maximum \(W_M(r)\). Then there exists \(\lambda_g(r)\in\mathbb{R}\) such that
\[a_M(q;r)-c_F(q)\le\lambda_g(r)\ \forall q\in X,\quad\text{and}\quad a_M(q;r)-c_F(q)=\lambda_g(r)\ \forall q\in\supp\left(dF\right).\]
\end{proposition}

\begin{proof}
Fix opponents at \(M\) and fix the deviator's price at \(r\). At that price, the deviator optimizes \(\int_X a_M(q;r) dH(q)-C(H)\) over \(H\in\mathcal{F}\). Since \(F\) attains this maximum, apply \Cref{lem:game-kkt}.
\end{proof}
Immediately, we have a mixed-strategy analog of \Cref{cor:price_equals_mc_gap}:
\begin{corollary}\label{cor:mixed_s7:73}
Posit \Cref{ass:cost,as:price_smoothing}. Fix \(M\) and \(r\) as in \Cref{prop:mixed_s7:FOC}, and let \(F\in\mathcal{F}\) attain \(W_M(r)\). Define the support endpoints \(\ubar{q}\coloneqq\inf\supp\left(dF\right)\) and \(\bar q\coloneqq\sup\supp\left(dF\right)\); and define, for each \(q\in X\),
\[\omega_M(q;r)\coloneqq\int_0^1\left(\tilde F_M\left((1-\sigma)q+\sigma e-r\right)\right)^{n-1}dG(e).\]
Then \(c_F\left(\bar q\right)-c_F\left(\ubar{q}\right)=r\cdot\left(\omega_M\left(\bar q;r\right)-\omega_M\left(\ubar{q};r\right)\right)\) and so \(0\le c_F\left(\bar q\right)-c_F\left(\ubar{q}\right)\le r\).
\end{corollary}

\begin{proof}
By \Cref{prop:mixed_s7:FOC}, for every \(q\in\supp\left(dF\right)\), \(r\omega_M(q;r)-c_F(q)=\lambda_g(r)\). Since \(\ubar{q},\bar q\in\supp\left(dF\right)\), subtracting the equality at \(\ubar{q}\) from that at \(\bar q\) yields the identity. Finally, \(\tilde F_M(\cdot)\in[0,1]\) implies \(0\le\omega_M\left(\bar q;r\right)-\omega_M\left(\ubar{q};r\right)\le1\), proving the bound.
\end{proof}

Suppose \(F\) is deterministic, but accompanied by random prices.

\begin{corollary}\label{cor:mixed_s7:detF}
Posit \Cref{ass:cost,as:price_smoothing}. Suppose a symmetric equilibrium is of the form \(\left(F,S\right)\), where \(F\in\mathcal{F}\) is deterministic and \(S\) is a cdf on \(P\). Let \(\ubar{p}\coloneqq\inf\supp\left(dS\right)\). Then, \(c_F\left(\bar q\right)-c_F\left(\ubar{q}\right)\le\ubar{p}\).
\end{corollary}

\begin{proof}
Take any \(p\in\supp\left(dS\right)\), so that \((F,p)\) must be a best response to opponents' mixed strategies. Apply \Cref{cor:mixed_s7:73} with \(r=p\). Taking the infimum over \(p\in\supp\left(dS\right)\) yields the inequality \(c_F\left(\bar q\right)-c_F\left(\ubar{q}\right)\le\ubar{p}\).\end{proof} This delivers the desired large-market outcome:
\begin{corollary}\label{cor:mixed_s7:detF_convergence}
Posit \Cref{ass:cost,as:price_smoothing}. For each \(n\ge2\), let \(\left(F_n,S_n\right)\) be a symmetric equilibrium of the \(n\)-firm game in which firms choose a deterministic \(F_n\in\mathcal{F}\) and mix over prices via \(S_n\). Then,\footnote{Our notation carries over in the obvious way: \(\ubar{q}_n\coloneqq\inf\supp\left(dF_n\right)\), \(\bar q_n\coloneqq\sup\supp\left(dF_n\right)\), and \(\ubar{p}_n\coloneqq\inf\supp\left(dS_n\right)\).}
\begin{enumerate}[noitemsep]
\item If \(\ubar{p}_n\to0\) as \(n\to\infty\), then \(c_{F_n}\left(\bar q_n\right)-c_{F_n}\left(\ubar{q}_n\right)\to0\).
\item \label{itlastcorr2} Moreover, if there exists \(\kappa>0\) such that for all \(F\in\mathcal{F}\) and all \(0\le x<y\le1\), \(c_F(y)-c_F(x)\ge\kappa(y-x)\),  then \(\ubar{p}_n\to0\) \(\Longrightarrow\) \(\bar q_n-\ubar{q}_n\to0\).
\end{enumerate}
\end{corollary}

\begin{proof}
\Cref{cor:mixed_s7:detF} implies \(c_{F_n}\left(\bar q_n\right)-c_{F_n}\left(\ubar{q}_n\right)\le\ubar{p}_n\), so \(\ubar{p}_n\to0\) implies the first claim. For the second, the steepness condition implies \(c_{F_n}\left(\bar q_n\right)-c_{F_n}\left(\ubar{q}_n\right)\ge\kappa\left(\bar q_n-\ubar{q}_n\right)\), so the first claim yields \(\bar q_n-\ubar{q}_n\to0\).
\end{proof}

Finally, we derive the result when there is joint mixing over prices and quality distributions.

\begin{corollary}\label{cor:mixed_s7:joint}
Posit \Cref{ass:cost,as:price_smoothing}. For each \(n\ge2\), let \(M_n\) be a symmetric equilibrium mixed strategy on \(\mathcal{F}\times P\). For any \((F,p)\in\supp\left(M_n\right)\), define \(\ubar{q}(F)\coloneqq\inf\supp\left(dF\right)\) and \(\bar q(F)\coloneqq\sup\supp\left(dF\right)\). Then, if
\[\bar p_n\coloneqq\sup\left\{p\in P\colon \exists F\in\mathcal{F}\text{ with }(F,p)\in\supp\left(M_n\right)\right\}\to0\quad\text{as }n\to\infty,\]
then \(\sup_{(F,p)\in\supp\left(M_n\right)}\left(c_F\left(\bar q(F)\right)-c_F\left(\ubar{q}(F)\right)\right)\to0\). Moreover, under the same uniform steepness condition as in \Cref{cor:mixed_s7:detF_convergence}(\ref{itlastcorr2}), \(\sup_{(F,p)\in\supp\left(M_n\right)}\left(\bar q(F)-\ubar{q}(F)\right)\to0\).
\end{corollary}

\begin{proof}
Fix \(n\) and \((F,p)\in\supp\left(M_n\right)\). Consequently, as payoffs are continuous, \((F,p)\) must be a best response to opponents' mixed strategy \(M_n\). Apply \Cref{cor:mixed_s7:73} with \(M=M_n\) and \(r=p\) to obtain \(c_F\left(\bar q(F)\right)-c_F\left(\ubar{q}(F)\right)\le p\). Taking the supremum over \((F,p)\in\supp\left(M_n\right)\) yields \(\sup_{(F,p)\in\supp\left(M_n\right)}\left(c_F\left(\bar q(F)\right)-c_F\left(\ubar{q}(F)\right)\right)\le\bar p_n\), which implies the stated convergence when \(\bar p_n\to0\). Finally, the steepness condition delivers, for each \((F,p)\in\supp\left(M_n\right)\), \(c_F\left(\bar q(F)\right)-c_F\left(\ubar{q}(F)\right)\ge\kappa\left(\bar q(F)-\ubar{q}(F)\right)\), so the cost-gap convergence implies \(\sup_{(F,p)\in\supp\left(M_n\right)}\left(\bar q(F)-\ubar{q}(F)\right)\to0\).
\end{proof}

\end{document}